\pdfoutput=1
\documentclass[
affil-it, 
auth-lg, 
british, 
compression, 
contents, 
custom-packages={scalerel}, 
references={bigone}, 
]{lib/preprint}

\newcommand{\Sh}{\overline{\mathrm{Sh}}}

\begin{document}
    
    \hypersetup{
        pdftitle = {Field Theory Equivalences as Spans of L infinity algebras},
        pdfauthor = {Mehran Jalali Farahani,Christian Saemann,Martin Wolf},
        pdfkeywords = {},
    }
    
    \date{\today}
    
    \email{mj2020@hw.ac.uk,c.saemann@hw.ac.uk,m.wolf@surrey.ac.uk}
    
    \preprint{EMPG--23--08,DMUS--MP--23/10}
    
    \title{Field Theory Equivalences as Spans of $L_\infty$-algebras} 
    
    \author[a]{Mehran~Jalali~Farahani\,\orcidlink{0009-0002-8282-9316}\,}
    \author[a]{Christian~Saemann\,\orcidlink{0000-0002-5273-3359}\,}
    \author[b]{Martin~Wolf\,\orcidlink{0009-0002-8192-3124}\,}
    
    \affil[a]{Maxwell Institute for Mathematical Sciences, Department of Mathematics,\\ Heriot--Watt University, Edinburgh EH14 4AS, United Kingdom}
    \affil[b]{School of Mathematics and Physics,\\ University of Surrey, Guildford GU2 7XH, United Kingdom}
    
    \abstract{Semi-classically equivalent field theories are related by a quasi-isomorphism between their underlying $L_\infty$-algebras, but such a quasi-isomorphism is not necessarily a homotopy transfer. We demonstrate that all quasi-isomorphisms can be lifted to spans of $L_\infty$-algebras in which the quasi-isomorphic $L_\infty$-algebras are obtained from a correspondence $L_\infty$-algebra by a homotopy transfer. Our construction is very useful: homotopy transfer is computationally tractable, and physically, it amounts to integrating out fields in a Feynman diagram expansion. Spans of $L_\infty$-algebras allow for a clean definition of quasi-isomorphisms of cyclic $L_\infty$-algebras. Furthermore, they appear naturally in many contexts within physics. As examples, we first consider scalar field theory with interaction vertices blown up in different ways. We then show that (non-Abelian) T-duality can be seen as a span of $L_\infty$-algebras, and we provide full details in the case of the principal chiral model. We also present the relevant span of $L_\infty$-algebras for the Penrose--Ward transform in the context of self-dual Yang--Mills theory and Bogomolny monopoles.}
    
    \acknowledgements{We would like to thank Leron Borsten, Branislav Jur\v{c}o, Hyungrok Kim, Jan Pulmann, and Jim Stasheff for helpful conversations. We are particularly grateful to Jan Pulmann for pointing out a gap in the proof of Proposition 2.6. M.J.F.~was supported by the STFC PhD studentship ST/W507489/1. C.S.~was partially supported by the Leverhulme Research Project Grant RPG-2018-329.}
    
    \datalicencemanagement{No additional research data beyond the data presented and cited in this work are needed to validate the research findings in this work. For the purpose of open access, the authors have applied a Creative Commons Attribution (CC BY) licence to any Author Accepted Manuscript version arising.}  
    
    \begin{body}
        
        \section{Introduction}
        
        There is a striking parallel between perturbative (quantum) field theory and homotopical algebra. This parallel is made evident by the Batalin--Vilkovisky (BV) formalism~\cite{Batalin:1977pb,Batalin:1981jr,Batalin:1984jr,Batalin:1984ss,Batalin:1985qj,Schwarz:1992nx}, which produces a differential graded commutative algebra, the BV complex. This complex, in turn, is dual to a homotopy algebraic structure called cyclic $L_\infty$-algebra, cf.~e.g.~\cite{Alexandrov:1995kv,Stasheff:1997iz,Zeitlin:2007fp,Jurco:2018sby}. More precisely, it is the Chevalley--Eilenberg algebra of the cyclic $L_\infty$-algebra. An $L_\infty$-algebra is a generalisation of a differential graded Lie algebra in which the Jacobi identity is violated up to homotopies, resulting in a tower of homotopy Jacobi identities. In physics, these homotopy Jacobi identities amount to closure of gauge transformations and gauge covariance of the equation of motion.
        
        All perturbative ghosts, fields, anti-fields, anti-fields of ghosts, etc., arrange into a graded vector space, and the free or linear terms in the equations of motion of the BV action give rise to differentials, turning the graded vector space into a cochain complex. We note that the cohomology of this cochain complex is given by the free fields up to gauge transformations. The interaction terms in the equations of motion that are of order $n$ in the fields define operations with $n$ inputs and one output, which provide the higher products of the $L_\infty$-algebra. The additional structures (inner products and integrals) contained in the BV action over its equations of motion induce a metric structure on the $L_\infty$-algebra. These facts have been observed and explored further numerous times since the birth of $L_\infty$-algebras in the context of closed string field theory~\cite{Zwiebach:1992ie}, see~e.g.~\cite{Kajiura:2003ax} and~\cite{Doubek:2017naz} for important examples and~\cite{Jurco:2018sby,Jurco:2019bvp} for a more complete list of references. Well-known is also~\cite{Hohm:2017pnh}, but in this paper the link to the BV formalism seems to have been made only partially.
        
        Remarkably, the parallel between quantum field theory and homotopical algebra extends far beyond the equations of motion. Homotopy algebras come with a notion of quasi-isomorphism, extending the corresponding one from cochain complexes. Quasi-isomorphisms between $L_\infty$-algebras translate to semi-classically equivalent field theories, i.e.~field theories with the same tree-level scattering amplitudes. Moreover, any homotopy algebra possesses a quasi-isomorphic minimal model, i.e.~a homotopy algebraic structure on the cohomology of their underlying cochain complex, which is unique up to isomorphisms. In the case of $L_\infty$-algebras of field theories, the minimal model encodes the tree-level scattering amplitudes. Furthermore, it is conveniently computed by the homological perturbation lemma~\cite{brown1967twisted,Gugenheim1989:aa,gugenheim1991perturbation,Crainic:0403266}, which encodes the usual tree-level Feynman diagram expansion in a geometric series. This geometric series gives rise to Berends--Giele recursion relations, see e.g.~\cite{Macrelli:2019afx,Jurco:2019yfd}. As already implied in~\cite{Zwiebach:1992ie}, much of this extends to the loop level, see~\cite{Markl:1997bj,Doubek:2017naz} as well as the closely related work~\cite{Costello:2016vjw,Costello:2021jvx}.
        
        Generally, the homological perturbation lemma may be used to extend any homotopy retract (i.e.~a weaker form of a homotopy equivalence) between the cochain complex of an $L_\infty$-algebra and another cochain complex to a quasi-isomorphism of $L_\infty$-algebras. This is known as homotopy transfer, see~\cite{Loday:2012aa} for a detailed account. From a field theoretic perspective, a homotopy transfer translates a field theory on a field space to an equivalent field theory on another field space. If the latter space is embedded in the former, then a homotopy transfer amounts to integrating out fields, a well-known fact in BV quantisation\footnote{See also~\cite{Dresse:1990dj,Henneaux:1989ua,Barnich:2004cr} for work on homotopy transfer from the BV perspective.}. For a recent discussion and application of this fact, see~\cite{Arvanitakis:2020rrk,Arvanitakis:2021ecw}.
        
        However, not all equivalences between field theories or, equivalently, quasi-isomorphisms of $L_\infty$-algebras can be captured by a homotopy transfer.\footnote{The mathematical question of which homotopy algebras arise from a homotopy transfer was originally raised by Sullivan and answered comprehensively in~\cite{Markl:2006.00072}.} From a field theoretic perspective, this observation is hardly surprising. For example, there is a quasi-isomorphism between the $L_\infty$-algebra that describes the tree-level scattering amplitudes of a field theory and the $L_\infty$-algebra that describes the action of this field theory. However, we clearly cannot reconstruct the complete form of a field theory from its tree-level scattering amplitudes.\footnote{that is, without invoking further constraints, such as locality} This, however, is unfortunate, as the perturbative expansions in terms of Feynman diagrams implied by homotopy transfer, together with the underlying recursion relations, can be very useful. 
        
        There are therefore many situations in which a physically equivalent field theory is constructed in a two-step-procedure, by first integrating in fields and then integrating out different fields. Notable examples are (non-Abelian) T-duality for sigma models and the Penrose--Ward transform. This naturally suggests the picture that for any two semi-classically equivalent field theories $\frL^{(1)}$ and $\frL^{(2)}$, there is a `correspondence theory' $\frL^{(c)}$ that is semi-classically equivalent to both $\frL^{(1)}$ and $\frL^{(2)}$ together with homotopy transfers between $\frL^{(c)}$ and $\frL^{(1)}$ and between $\frL^{(c)}$ and $\frL^{(2)}$, respectively, amounting to integrating out fields in $\frL^{(c)}$:
        \begin{equation}
            \begin{tikzcd}
                & \frL^{(c)} \arrow[<->,dl] \arrow[<->,dr] &
                \\
                \frL^{(1)} & & \frL^{(2)} 
            \end{tikzcd}
        \end{equation}
        Mathematically speaking, one may therefore conjecture that for any pair of $L_\infty$-algebras\footnote{The evident generalisation to arbitrary homotopy algebras should be straightforward.} $\frL^{(1)}$ and $\frL^{(2)}$, there is an $L_\infty$-algebra $\frL^{(c)}$ quasi-isomorphic to both $\frL^{(1)}$ and $\frL^{(2)}$, and where the quasi-isomorphism can be captured by a homotopy transfer. We note that the conjectured statement would also allow for a so far missing, clean definition of quasi-isomorphisms of cyclic $L_\infty$-algebras.
        
        We shall prove that this is indeed the case in \cref{sec:spans_of_l_infty_algebras}. We then work out the details for three distinct cases: firstly, an illustrative toy example of two quasi-isomorphic scalar field theories in \cref{sec:toy_model}; secondly, the much more intricate and interesting case of the principal chiral model and its (non-Abelian) T-dual in \cref{sec:PCM}; and thirdly, the Penrose--Ward transforms between field theories allowing for a twistorial description and holomorphic Chern--Simons theory on the corresponding twistor spaces in \cref{sec:twistors}. We note that further examples of such spans, but described from the dual, BV perspective, can be found in~\cite{Barnich:2010sw,Grigoriev:2012xg}.
        
        \section{Quasi-isomorphisms and homotopy transfer}\label{sec:spans_of_l_infty_algebras}
        
        Quasi-isomorphisms of cochain complexes, such as de Rham complexes, are cochain maps that induce isomorphisms between the underlying cohomology groups. If one of the cochain complexes carries a homotopy algebra structure, then this structure can be transferred to the other cochain complex in a procedure called homotopy transfer, see e.g~\cite{Kontsevich:2000er} and~\cite{Loday:2012aa} for a comprehensive review. Explicit formulas for such a homotopy transfer are provided by the homological perturbation lemma~\cite{brown1967twisted,Gugenheim1989:aa,gugenheim1991perturbation}, see also~\cite{Crainic:0403266}.
        
        However, not all quasi-isomorphisms of homotopy algebras originate from a homotopy transfer~\cite{Markl:2006.00072}. As we show below for the case of $L_\infty$-algebras, however, each quasi-isomorphism can be lifted to a \uline{span} (or roof, or correspondence) of homotopy algebras, in which the projections are given by homotopy transfers. This perspective is very natural for a number of reasons. In particular, $L_\infty$-algebras concentrated in non-positive degrees integrate to Lie $\infty$-groupoids. Already for ordinary groupoids, fully faithful and essentially surjective functors can fail to be equivalences. In such situations, one usually switches to spans.
        
        We will show that our new perspective gives rise to a natural definition of quasi-isomorphism between cyclic $L_\infty$-algebras, which so far existed only in special cases. Moreover, it is useful in the context of perturbative quantum field theory, where quasi-isomorphisms of $L_\infty$-algebras amount to a semi-classical equivalence, and homotopy transfers amount to Feynman diagram expansions arising from integrating out fields.
        
        \subsection{Quasi-isomorphisms and homotopy transfer for \texorpdfstring{$L_\infty$}{Linfty}-algebras}
        
        \paragraph{$L_\infty$-algebras.}
        In the following, we shall restrict ourselves to strong homotopy Lie algebras, or $L_\infty$-algebras for short. These can be defined as differential graded cocommutative coalgebras, or, dually, as differential graded commutative algebras, see e.g.~\cite{Jurco:2018sby} for a review in our conventions and~\cite{Zwiebach:1992ie,Lada:1992wc,Lada:1994mn} for original literature. Explicitly, we shall think of an $L_\infty$-algebra as a $\IZ$-graded vector space\footnote{This definition has an evident generalisation to graded modules.}
        \begin{subequations}
            \begin{equation}
                \frL\ =\ \bigoplus_{k\in\IZ}\frL_k
            \end{equation}
            endowed with multilinear totally antisymmetric \uline{higher products} $\mu_i$ of degree~$2-i$ for all $i\in\IN$,
            \begin{equation}
                \mu_i\,:\,\frL^{\times i}\ \rightarrow\ \frL~,
            \end{equation}
            subject to the \uline{homotopy Jacobi identities}
            \begin{equation}\label{eq:homotopyJacobi}
                \sum_{j+k=i}\sum_{\sigma\in \Sh(j;i) }\chi(\sigma;a_1,\ldots,a_i)(-1)^{k}\mu_{k+1}(\mu_j(a_{\sigma(1)},\ldots,a_{\sigma(j)}),a_{\sigma(j+1)},\ldots,a_{\sigma(i)})\ =\ 0
            \end{equation}
            for all $i\in\IN$ and $a_1,\ldots,a_{i}$ homogeneous elements in $\frL$. Here, $\Sh(j;i)$ denotes the set of $(j;i)$-unshuffles, i.e.~permutations $\sigma$ of $\{1,\ldots,i\}$ such that
            \begin{equation}
                \sigma(1)\ <\ \cdots\ <\ \sigma(j)
                \eand
                \sigma(j+1)\ <\ \cdots\ <\ \sigma(i)~,
            \end{equation}
            and $\chi(\sigma;a_1,\ldots,a_i)$ is the graded Koszul sign for permutations of homogeneous elements, defined by
            \begin{equation}\label{eq:DefKoszulSign}
                a_1\wedge\ldots\wedge a_i\ =\ \chi(\sigma;a_1,\ldots,a_i)\,a_{\sigma(1)}\wedge\ldots\wedge a_{\sigma(i)}~.
            \end{equation}
            We note that $\mu_1$ is a differential, and $(\frL,\mu_1)$ forms a cochain complex.
            
            An $L_\infty$-algebra $\frL$ is called cyclic, if there is an additional symmetric, non-degenerate, bilinear form
            \begin{equation}
                \inner{-}{-}\,:\,\frL\times\frL\ \rightarrow\ \IK
            \end{equation}
            with $\IK$ the ground field or ring such that
            \begin{equation}
                \inner{a_1}{\mu_i(a_2,\ldots,a_{i+1})}\ =\ (-1)^{i+i(|a_1|+|a_{i+1}|)+|a_{i+1}|\sum_{j=1}^i|a_j|}\inner{a_{i+1}}{\mu_i(a_1,\ldots,a_i)}~.
            \end{equation}
            We shall refer to this as a \uline{metric structure}.
        \end{subequations}
        
        \paragraph{Morphisms of $L_\infty$-algebras.} 
        Let $(\frL^{(1)},\mu_i^{(1)})$ and $(\frL^{(2)},\mu_i^{(2)})$ be two $L_\infty$-algebras. Strict morphisms of $L_\infty$-algebras $\phi:\frL^{(1)}\rightarrow \frL^{(2)}$ are cochain maps between the underlying cochain complexes that respect the higher brackets,
        \begin{equation}
            \mu^{(2)}_i(\phi(a_1),\ldots,\phi(a_i))\ =\ \phi(\mu^{(1)}_i(a_1,\ldots,a_i))
        \end{equation}
        for all $i\in\IN$ and $a_1,\ldots,a_i\in\frL^{(1)}$. 
        
        More generally, $L_\infty$-algebras can be regarded as codifferential graded cocommutative coalgebras, and \uline{(weak) morphisms of $L_\infty$-algebras} amount to morphisms between these. A general such morphism $\phi:\frL^{(1)}\rightarrow\frL^{(2)}$ consists of a collection of multilinear totally anti-symmetric maps $\phi_i$ of degree~$1-i$, 
        \begin{subequations}\label{eq:L_infty_morphism}
            \begin{equation}
                \phi_i\,:\,(\frL^{(1)})^{\times i}\ \rightarrow\ \frL^{(2)}~,
            \end{equation}
            which relate the higher products $\mu_i^{(1)}$ and $\mu_i^{(2)}$ of $\frL^{(1)}$ and $\frL^{(2)}$ according to
            \begin{equation}
                \begin{aligned}
                    &\sum_{j+k=i}\sum_{\sigma\in \Sh(j;i)}~(-1)^{k}\chi(\sigma;a_1,\ldots,a_i)\phi_{k+1}(\mu^{(1)}_j(a_{\sigma(1)},\ldots,a_{\sigma(j)}),a_{\sigma(j+1)},\ldots,a_{\sigma(i)})\\
                    \ &=\ \sum_{j=1}^i\frac{1}{j!} \sum_{k_1+\cdots+k_j=i}\sum_{\sigma\in\Sh(k_1,\ldots,k_{j-1};i)}\chi(\sigma;a_1,\ldots,a_i)\zeta(\sigma;a_1,\ldots,a_i)\,\times\\
                    &\kern1cm\times \mu^{(2)}_j\Big(\phi_{k_1}\big(a_{\sigma(1)},\ldots,a_{\sigma(k_1)}\big),\ldots,\phi_{k_j}\big(a_{\sigma(k_1+\cdots+k_{j-1}+1)},\ldots,a_{\sigma(i)}\big)\Big)~,
                \end{aligned}
            \end{equation}
            where
            \begin{equation}\label{eq:zeta}
                \zeta(\sigma;a_1,\ldots,a_i)\ \coloneqq\ (-1)^{\sum_{1\leq m<n\leq j}k_mk_n+\sum_{m=1}^{j-1}k_m(j-m)+\sum_{m=2}^j(1-k_m)\sum_{k=1}^{k_1+\cdots+k_{m-1}}|a_{\sigma(k)}|}~,
            \end{equation}
            and $\Sh(k_1,\ldots,k_{j-1};i)$ denotes the set of generalised $(k_1,\ldots,k_{j-1};i)$-unshuffles, i.e.~permutations for which the first $k_1$ images, the next $k_2$ images, etc.~are all ordered. 
        \end{subequations}
        
        For illustrative purposes, let us list the relation between the lowest three higher products explicitly,
        \begin{equation}\label{eq:explicit_formulas}
            \begin{aligned}
                &\mu^{(2)}_1(\phi_1(a_1))\ =\ \phi_1\big(\mu^{(1)}_1(a_1)\big)~,
                \\
                &\mu^{(2)}_2(\phi_1(a_1),\phi_1(a_2))\ =\ \phi_1\big(\mu^{(1)}_2(a_1,a_2)\big)-\phi_2\big(\mu^{(1)}_1(a_1),a_2\big)+(-1)^{|a_1|\,|a_2|}\phi_2\big(\mu^{(1)}_1(a_2),a_1\big)
                \\
                &\hspace{1cm}-\mu^{(2)}_1(\phi_2(a_1,a_2))~,
                \\
                &\mu^{(2)}_3(\phi_1(a_1),\phi_1(a_2),\phi_1(a_3))\ =\ (-1)^{|a_2|\,|a_3|}\phi_2\big(\mu^{(1)}_2(a_1,a_3),a_2\big)
                \\
                &\hspace{1cm}+(-1)^{|a_1|(|a_2|+|a_3|)+1}\phi_2\big(\mu^{(1)}_2(a_2,a_3),a_1)+(-1)^{|a_1|\,|a_2|+1}\phi_3\big(\mu^{(1)}_1(a_2),a_1,a_3\big)
                \\
                &\hspace{1cm}+(-1)^{(|a_1|+|a_2|)|a_3|} \phi_3\big(\mu^{(1)}_1(a_3),a_1,a_2\big)+(-1)^{|a_1|}\mu^{(2)}_2(\phi_1(a_1),\phi_2(a_2,a_3))
                \\
                &\hspace{1cm}
                -(-1)^{(|a_1|+1)|a_2|}\mu^{(2)}_2(\phi_1(a_2),\phi_2(a_1,a_3))+(-1)^{(|a_1|+|a_2|+1)|a_3|}\mu^{(2)}_2(\phi_1(a_3),\phi_2(a_1,a_2))
                \\
                &\hspace{1cm}
                +\phi_1\big(\mu^{(1)}_3(a_1,a_2,a_3)\big)-\phi_2\big(\mu^{(1)}_2(a_1,a_2),a_3\big)+\phi_3\big(\mu^{(1)}_1(a_1),a_2,a_3\big)
                \\
                &\hspace{1cm}-\mu^{(2)}_1(\phi_3(a_1,a_2,a_3))
            \end{aligned}
        \end{equation}
        for all $a_1,\ldots,a_3\in \frL^{(1)}$. Evidently, $\phi_1$ is a cochain map on the underlying cochain complex. Moreover, a weak morphism with $\phi_i$ trivial for $i\geq 2$ is a strict morphism.
        
        A \uline{quasi-isomorphism} between two $L_\infty$-algebras is now a morphism of $L_\infty$-algebras such that $\phi_1$ induces an isomorphism between the cohomologies of the corresponding cochain complexes. Quasi-isomorphic $L_\infty$-algebras can be regarded as equivalent for most intents and purposes.
        
        We say that a weak morphisms $\phi:\frL^{(1)}\rightarrow\frL^{(2)}$ respects the metric structure if it satisfies~\cite{Kajiura:2003ax}
        \begin{equation}\label{eq:compatibility_metric_structure}
            \begin{gathered}
                \inner{\phi_1(a_1)}{\phi_1(a_2)}^{(2)}\ =\ \inner{a_1}{a_2}^{(1)}~,
                \\
                \sum_{\substack{j+k=i\\j,k\geq 1}}\inner{\phi_j(a_1,\ldots,a_j)}{\phi_k(a_{j+1},\ldots,a_{j+k}))}^{(2)}\ =\ 0
            \end{gathered}
        \end{equation}
        for all $i\geq 3$ and $a_1,\ldots,a_i\in\frL^{(1)}$. We note that this condition, together with the non-degeneracies of the metric structures, implies that $\phi_1$ is injective.
        
        \paragraph{Homotopy transfer.}
        A \uline{deformation retract} (see e.g.~\cite{Loday:2012aa}) between two cochain complexes $(\frL^{(1)},\mu_1^{(1)})$ and $(\frL^{(2)},\mu_1^{(2)})$ is a pair of cochain maps 
        $\sfp$ and $\sfe$ together with a linear map $\sfh$ of degree~$-1$ that fit into the diagram 
        \begin{subequations}\label{eq:deformation_retract}
            \begin{equation}
                \begin{tikzcd}
                    \ar[loop,out=160,in=200,distance=20,"\sfh" left] (\frL^{(1)},\mu_1^{(1)})\arrow[r,shift left]{}{\sfp} & (\frL^{(2)},\mu^{(2)}_1)\arrow[l,shift left]{}{\sfe}
                \end{tikzcd}
            \end{equation}
            and satisfy
            \begin{equation}\label{eq:deformation_retract_2}
                \id_{\frL^{(1)}}-\sfe\circ\sfp\ =\ \mu_1^{(1)}\circ\sfh+\sfh\circ\mu_1^{(1)}\eand
                \sfp\circ\sfe\ =\ \sfid_{\frL^{(2)}}~.
            \end{equation} 
        \end{subequations}       
        It is thus a stricter form of a homotopy equivalence between cochain complexes. A deformation retract is called \uline{special} if also the so-called annihilation or \uline{side conditions} are satisfied,
        \begin{equation}\label{eq:HT_side_conditions}
            \sfp\circ\sfh\ =\ 0~,
            \quad
            \sfh\circ\sfe\ =\ 0~,
            \eand
            \sfh\circ\sfh\ =\ 0~.
        \end{equation}
        A deformation retract can always be turned into a special deformation retract, cf.~\cite{Crainic:0403266} or~\cite{Loday:2012aa}, by performing the replacement
        \begin{equation}\label{eq:special_substitutions}
            \sfh\ \rightarrow \ (\sfid-\sfe\circ\sfp)\circ\sfh\circ(\sfid-\sfe\circ\sfp)\circ\mu^{(1)}_1\circ(\sfid-\sfe\circ\sfp)\circ\sfh\circ(\sfid-\sfe\circ\sfp)~.
        \end{equation}
        In the following, we shall always work with special deformation retracts. 
        
        Given a special deformation retract as in~\eqref{eq:deformation_retract} and an $L_\infty$-algebra structure on $\frL^{(1)}$, the \uline{homological perturbation lemma} allows us to transfer this structure to the cochain complex $\frL^{(2)}$, and detailed formulas\footnote{albeit for $A_\infty$-algebras, but readily translatable} are found e.g.~in~\cite{Kajiura:2003ax,Markl:0401007}. In principle, one can now consider the homological perturbation lemma for cyclic $L_\infty$-algebras, as done e.g.~in~\cite{Doubek:2017naz}. For our purposes, however, it will be easier to work with morphisms of general $L_\infty$-algebras and check that the result we obtain is cyclic.
        
        Explicitly, the homological perturbation lemma provides a lift of the quasi-isomorphism of cochain complexes $\sfe$ to a quasi-isomorphism of $L_\infty$-algebras $\sfE:\frL^{(2)}\rightarrow\frL^{(1)}$. In particular, given an $L_\infty$-structure on $\frL^{(1)}$, it induces an $L_\infty$-structure on $\frL^{(2)}$ via the component maps
        \begin{subequations}\label{eq:transfer_formulas}
            \begin{equation}\label{eq:QuasiIsomorphismInTransfer}
                \begin{aligned}
                    \sfE_1(b_1)\ &\coloneqq\ \sfe(b_1)~,
                    \\
                    \sfE_2(b_1,b_2)\ &\coloneqq\ -\sfh(\mu^{(1)}_2(\sfE_1(b_1),\sfE_1(b_2)))~,
                    \\
                    &~~\vdots
                    \\
                    \sfE_i(b_1,\ldots,b_i)\ &\coloneqq\ -\sum_{j=2}^i\frac{1}{j!} \sum_{\substack{k_1+\cdots+k_j=i\\k_1,\ldots,k_j\geq 1}}\sum_{\sigma\in{\rm Sh}(k_1,\ldots,k_{j-1};i)}\chi(\sigma;b_1,\ldots,b_i)\zeta(\sigma;b_1,\ldots,b_i)
                    \\
                    &\kern0.5cm\times\sfh\left\{\mu_j^{(1)}\Big(\sfE_{k_1}\big(b_{\sigma(1)},\ldots,b_{\sigma(k_1)}\big),\ldots,\sfE_{k_j}\big(b_{\sigma(k_1+\cdots+k_{j-1}+1)},\ldots,b_{\sigma(i)}\big)\Big)\right\}
                \end{aligned}
            \end{equation}
            for all $b_1,\ldots,b_i\in\frL^{(2)}$ so that the induced higher products on $\frL^{(2)}$ are given by 
            \begin{equation}\label{eq:TransferredHigherProducts}
                \begin{aligned}
                    \mu^{(2)}_2(b_1,b_2)\ &\coloneqq\ \sfp(\mu^{(1)}_2(\sfE_1(b_1),\sfE_1(b_2))~,
                    \\
                    &~~\vdots
                    \\
                    \mu^{(2)}_i(b_1,\ldots,b_i)\ &\coloneqq\ \sum_{j=2}^i\frac{1}{j!} \sum_{\substack{k_1+\cdots+k_j=i\\k_1,\ldots,k_j\geq 1}}\sum_{\sigma\in{\rm Sh}(k_1,\ldots,k_{j-1};i)}\chi(\sigma;b_1,\ldots,b_i)\zeta(\sigma;b_1,\ldots,b_i)
                    \\
                    &\kern0.5cm\times\sfp\left\{\mu_j^{(1)}\Big(\sfE_{k_1}\big(b_{\sigma(1)},\ldots,b_{\sigma(k_1)}\big),\ldots,\sfE_{k_j}\big(b_{\sigma(k_1+\cdots+k_{j-1}+1)},\ldots,b_{\sigma(i)}\big)\Big)\right\}
                \end{aligned}
            \end{equation}
            for all $b_1,\ldots,b_i\in\frL^{(2)}$ with the sign $\zeta$ as defined in~\eqref{eq:zeta}.
        \end{subequations}        
        
        \paragraph{Quasi-isomorphisms not originating from homotopy transfer.}
        Upon comparing the formulas~\eqref{eq:transfer_formulas} with those for a general quasi-isomorphism~\eqref{eq:explicit_formulas}, we can straightforwardly identify quasi-isomorphisms that do not originate from a homotopy transfer.
        
        As an example, consider the trivial, one-element $L_\infty$-algebra $\frL^{(1)}$ with the underlying cochain complex
        \begin{equation}
            \sfCh(\frL^{(1)})\ \coloneqq\ \big(\cdots\xrightarrow{~~~}\{0\}\xrightarrow{~~~}\{0\}\xrightarrow{~~~}\{0\}\xrightarrow{~~~}\cdots\big)
        \end{equation}
        with no non-trivial higher products and the $L_\infty$-algebra $\frL^{(2)}$ with the underlying cochain complex
        \begin{equation}
            \sfCh(\frL^{(2)})\ \coloneqq\ \Big(\cdots\xrightarrow{~~~}\{0\}\xrightarrow{~~~}\underbrace{\frg}_{\frL^{(2)}_{-1}}\xrightarrow{~\sfid~}\underbrace{\frg}_{\frL^{(2)}_{0}}\xrightarrow{~~~}\{0\}\xrightarrow{~~~}\cdots\Big)
        \end{equation}
        for some Lie algebra $\frg$ with the only non-trivial higher product
        \begin{equation}
            \mu^{(2)}_2(b_1,b_2)\ \coloneqq\ [b_1,b_2]
        \end{equation}
        for all $b_1,b_2\in\frL^{(2)}$ and $|b_1|+|b_2|\geq -1$. Because the cohomologies $H^\bullet_{\mu_1^{(1)}}(\frL^{(1)})$ and $H^\bullet_{\mu_1^{(2)}}(\frL^{(2)})$ are both trivial, the trivial map from $\frL^{(1)}$ to $\frL^{(2)}$ is a quasi-isomorphism. The higher products on $\frL^{(2)}$, however, clearly do not originate from a homotopy transfer of the (trivial) higher products on $\frL^{(1)}$ and so, this quasi-isomorphism $\frL^{(1)}\rightarrow\frL^{(2)}$ is not given by a homotopy transfer. However, there certainly is a homotopy transfer $\frL^{(2)}\rightarrow\frL^{(1)}$.
        
        \paragraph{Minimal model.}
        Note that by the usual abstract Hodge--Kodaira decomposition, every $L_\infty$-algebra $\frL$ comes with a minimal model, that is, an $L_\infty$-algebra structure on its cohomology $H^\bullet_{\mu_1}(\frL)$, cf.~\cite{kadeishvili1982algebraic,Kajiura:2003ax}. This minimal model is obtained by homotopy transfer using the special deformation retract
        \begin{equation}\label{eq:MM_deformation_retract}
            \begin{tikzcd}
                \ar[loop,out=160,in=200,distance=20,"\sfh" left] (\frL,\mu_1)\arrow[r,shift left]{}{\sfp} & (H^\bullet_{\mu_1}(\frL),0) \arrow[l,shift left]{}{\sfe}~,
            \end{tikzcd}
        \end{equation}
        where $\sfe$ is an embedding of the cohomology $H^\bullet_{\mu_1}$ into $\frL$, $\sfp$ is a corresponding projection, and $\sfh$ is the contracting homotopy~\cite{Kontsevich:1997vb,Huebschmann:9906036}.
        
        \paragraph{Composition of homotopy transfers.}
        Note that two deformation retracts with matching source and target,
        \begin{equation}
            \begin{tikzcd}
                \ar[loop,out=160,in=200,distance=20,"\sfh^{(1)}" left] \big(\frL^{(1)},\mu_1^{(1)}\big)\arrow[r,shift left]{}{\sfp^{(21)}} & \big(\frL^{(2)},\mu^{(2)}_1\big)\arrow[l,shift left]{}{\sfe^{(12)}}
            \end{tikzcd}
            \eand
            \begin{tikzcd}
                \ar[loop,out=160,in=200,distance=20,"\sfh^{(2)}" left] \big(\frL^{(2)},\mu_1^{(2)}\big)\arrow[r,shift left]{}{\sfp^{(32)}} & \big(\frL^{(3)},\mu^{(3)}_1\big)\arrow[l,shift left]{}{\sfe^{(23)}}
            \end{tikzcd},
        \end{equation}
        can be combined into a single deformation retract
        \begin{subequations}\label{eq:combination_deformation_retracts}
            \begin{equation}
                \begin{tikzcd}
                    \ar[loop,out=160,in=200,distance=20,"\tilde\sfh^{(1)}" left] \big(\frL^{(1)},\mu_1^{(1)}\big)\arrow[r,shift left]{}{\sfp^{(31)}} & \big(\frL^{(3)},\mu^{(3)}_1\big) \arrow[l,shift left]{}{\sfe^{(13)}}
                \end{tikzcd}
            \end{equation}
            with 
            \begin{equation}
                \begin{gathered}
                    \sfp^{(31)}\ =\ \sfp^{(32)}\circ\sfp^{(21)}~,
                    \quad
                    \sfe^{(13)}\ =\ \sfe^{(12)}\circ\sfe^{(23)}~,
                    \\
                    \tilde\sfh^{(1)}\ =\ \sfh^{(1)}+\sfe^{(12)}\circ\sfh^{(2)}\circ\sfp^{(21)}~.
                \end{gathered}
            \end{equation}
        \end{subequations}        
        Indeed, we have
        \begin{subequations}
            \begin{equation}
                \sfp^{(31)}\circ\sfe^{(13)}\ =\ \sfp^{(32)}\circ\sfp^{(21)}\circ\sfe^{(12)}\circ\,\sfe^{(23)}\ =\ \sfp^{(32)}\circ\sfe^{(23)}\ =\ \id_{\frL^{(3)}}
            \end{equation}
            and likewise,
            \begin{equation}
                \begin{aligned}
                    \sfe^{(13)}_0\circ\sfp^{(31)}\ &=\ \sfe^{(12)}\circ\sfe^{(23)}\circ\sfp^{(32)}\circ\sfp^{(21)}
                    \\
                    &=\ \sfe^{(12)}\circ(\id_{\frL^{(2)}}-\mu_1^{(2)}\circ\sfh^{(2)}-\sfh^{(2)}\circ\mu_1^{(2)})\circ\sfp^{(21)}
                    \\
                    &=\ \sfe^{(12)}_0\circ\sfp^{(21)}-\mu_1^{(1)}\circ\sfe^{(12)}\circ\sfh^{(2)}\circ\sfp^{(21)}-\sfe^{(12)}\circ\sfh^{(2)}\circ\sfp^{(21)}\circ\mu_1^{(1)}
                    \\
                    &=\ \id_{\frL^{(1)}}-\mu_1^{(1)}\circ\big(\sfh^{(1)}+\sfe^{(12)}\circ\sfh^{(2)}\circ\sfp^{(21)}\big)-\big(\sfh^{(1)}+\sfe^{(12)}\circ\sfh^{(2)}\circ\sfp^{(21)}\big)\circ\mu_1^{(1)}
                    \\
                    &=\ \id_{\frL^{(1)}}-\mu_1^{(1)}\circ\tilde\sfh^{(1)}-\tilde\sfh^{(1)}\circ\mu_1^{(1)}~,
                \end{aligned}
            \end{equation}
        \end{subequations}
        where we have used that $\sfe$ and $\sfp$ are cochain maps.
        
        We can invert the above observation to the following result.
        \begin{proposition}\label{prop:factorize_contracting_homotopy}
            Given $L_\infty$-algebras $\frL^{(1,2,3)}$ together with projections and embeddings of the underlying cochain complexes
            \begin{equation}
                \begin{tikzcd}
                    \big(\frL^{(1)},\mu_1^{(1)}\big)\arrow[r,shift left]{}{\sfp^{(21)}} & \big(\frL^{(2)},\mu^{(2)}_1\big)\arrow[l,shift left]{}{\sfe^{(21)}}\arrow[r,shift left]{}{\sfp^{(32)}}
                    & \big(\frL^{(3)},\mu^{(3)}_1\big)\arrow[l,shift left]{}{\sfe^{(23)}}
                \end{tikzcd}
            \end{equation}
            such that
            \begin{equation}
                \sfp^{(21)}\circ \sfe^{(21)}\ =\ \sfid_{\frL^{(2)}}\eand \sfp^{(32)}\circ \sfe^{(23)}\ =\ \sfid_{\frL^{(3)}}
            \end{equation}
            and two special deformation retracts
            \begin{equation}
                \begin{gathered}
                    \begin{tikzcd}[column sep=2cm]
                        \ar[loop,out=160,in=200,distance=20,"\sfh^{(31)}" left] \big(\frL^{(1)},\mu_1^{(1)}\big)\arrow[r,shift left]{}{\sfp^{(32)}\circ\sfp^{(21)}} & \big(\frL^{(3)},\mu^{(3)}_1\big)\arrow[l,shift left]{}{\sfe^{(12)}\circ \sfe^{(23)}}
                    \end{tikzcd}~,
                    \\
                    \begin{tikzcd}
                        \ar[loop,out=160,in=200,distance=20,"\sfh^{(32)}" left] \big(\frL^{(2)},\mu_1^{(2)}\big)\arrow[r,shift left]{}{\sfp^{(32)}} & \big(\frL^{(3)},\mu^{(3)}_1\big)\arrow[l,shift left]{}{\sfe^{(23)}}
                    \end{tikzcd}~,
                \end{gathered}
            \end{equation}
            then there is a third special deformation retract
            \begin{equation}
                \begin{tikzcd}
                    \ar[loop,out=160,in=200,distance=20,"\sfh^{(21)}" left] \big(\frL^{(1)},\mu_1^{(1)}\big)\arrow[r,shift left]{}{\sfp^{(21)}} & \big(\frL^{(2)},\mu^{(2)}_1\big)\arrow[l,shift left]{}{\sfe^{(12)}}
                \end{tikzcd}
            \end{equation}
            with $\sfh^{(21)}$ obtained from the substitutions~\eqref{eq:special_substitutions} from 
            \begin{equation}\label{eq:h21}
                \sfh^{(21)}\ =\ \sfh^{(31)}-\sfe^{(12)}\circ \sfh^{(32)}\circ\sfp^{(21)}~.
            \end{equation}
        \end{proposition}
        \begin{proof}
            The required identity $\sfid_{\frL^{(1)}}=\mu_1^{(1)}\circ\sfh^{(21)}+\sfh^{(21)}\circ\mu_1^{(1)}+\sfe^{(12)}\circ\sfp^{(21)}$ for $\sfh^{(21)}$ of~\eqref{eq:h21} follows from direct computation.
        \end{proof}
        We now have the following, useful corollary.
        \begin{corollary}\label{cor:projections_are_homotopy_transfer}
            A strict projection $\sfp$ of an $L_\infty$-algebra $\hat\frL$ to a quasi-isomorphic $L_\infty$-subalgebra\footnote{i.e.~a vector subspace of $\hat\frL$ on which the higher products close} $\frL$ lifts to a homotopy transfer.
        \end{corollary}
        
        \begin{proof}
            Because $\frL$ is a subspace, we have besides the projection also an embedding:
            \begin{equation}
                \begin{tikzcd}
                    \big(\hat\frL,\hat\mu_1\big)\arrow[r,shift left]{}{\sfp} & \big(\frL,\mu_1\big)\arrow[l,shift left]{}{\sfe}
                \end{tikzcd}
            \end{equation}
            such that $\sfp\circ\sfe=\sfid_{\frL}$. We also have special deformation retracts from both $\hat\frL$ and $\frL$ to the joint minimal model $\frL^\circ$. By \cref{prop:factorize_contracting_homotopy}, we then also have a special deformation retract from $\hat \frL$ to $\frL$. Let us denote the induced higher products via homotopy transfer by $\mu_n^*$. We observe the identity $(\sfp\circ\hat\mu_k)(\alpha_1,\ldots,\alpha_k)=0$ if $\sfp(\alpha_i)=0$ for some $i$, because $\sfp$ is strict. In the resulting homotopy transfer, we have $\sfp\circ\sfE_j=0$, for $j>1$, where $\sfE_j$ is defined in~\eqref{eq:QuasiIsomorphismInTransfer}, because of the side conditions $\sfp\circ \sfh=0$. We then observe that~\eqref{eq:TransferredHigherProducts} reduces to            
            \begin{equation}
                \mu_k^*(b_1,\ldots,b_k)=(\sfp\circ \hat\mu_k)(\sfe(b_1),\ldots,\sfe(b_k))= \mu_k((\sfp\circ\sfe)(b_1),\ldots,(\sfp\circ\sfe)(b_k))=\mu_k(b_1,\ldots,b_k)~.
            \end{equation}   
            This means the higher products induced by homotopy transfer and the original higher products coincide.
        \end{proof}
        
        \subsection{Spans of \texorpdfstring{$L_\infty$}{Linfty}-algebras}
        
        It would certainly be useful if all computations of quasi-isomorphisms $L_\infty$-algebras were encoded in homotopy transfers as for those, explicit and recursive formulas are provided by the homological perturbation lemma. Moreover, in many applications to perturbative quantum field theory, it is useful to turn computations into Feynman diagram expansions with all their combinatorial features, which is essentially what the homological perturbation lemma does. In this section, we shall show that every quasi-isomorphism of $L_\infty$-algebras can indeed be encoded in pairs of homotopy transfers.
        
        \paragraph{Pullbacks of $L_\infty$-algebras.}
        Given two $L_\infty$-algebras $\frL^{(1)}$ and $\frL^{(2)}$ with surjections\footnote{i.e.~weak morphisms of $L_\infty$-algebras with the linear component $\sigma_1^{(1,2)}$ surjective} $\sigma^{(1,2)}$ to a third $L_\infty$-algebra $\frL^{(b)}$, then there is a fourth $L_\infty$-algebra $\frL^{(p)}$ that fits into the pullback diagram
        \begin{equation}\label{eq:first_pullback}
            \begin{tikzcd}
                \frL^{(p)}\arrow[d,""]\arrow[r]
                & 
                \frL^{(2)}\arrow[d,"\sigma^{(2)}"]
                \\
                \frL^{(1)}\arrow[r,"\sigma^{(1)}"]
                & 
                \frL^{(b)}
            \end{tikzcd}
        \end{equation}
        with the usual universality of $\frL^{(p)}$ arising in pullbacks. Abstractly, this is a consequence of the existence of pullbacks for homotopy algebras, cf.~\cite[Theorem 4.1]{Vallette:1411.5533}, and $\frL^{(p)}$ is called the pullback; see also~\cite{Rogers:1809.05999} for the special case of $L_\infty$-algebras concentrated in non-positive degrees. It remains to show that there exists an $L_\infty$-algebra $\frL^{(c)}$ quasi-isomorphic to $\frL^{(p)}$ such that there are homotopy transfers $\frL^{(c)}\rightarrow \frL^{(1,2)}$.
        
        \begin{remark}
            One may be tempted to think that the pullback~\eqref{eq:first_pullback} should be regarded as a homotopy pullback. This is not the case, as the $L_\infty$-algebras $\frL^{(1,2,b)}$ are all homotopically equivalent, and hence $\frL^{(p)}$ could be identified with $\frL^{(b)}$. 
            
            Also, one may think that the existence of a suitable $\frL^{(c)}$ follows trivially from the decomposition theorem, which states that any $L_\infty$-algebra $\frL$ decomposes as $\frL=\frL^\circ\oplus \frL^{\rm lc}$ into a minimal model and a linearly contractible\footnote{In a linearly contractible $L_\infty$-algebra, all higher products except for the differential vanish and the cohomology is trivial. These are known as ``trivial pairs'' by physicists.} one. It is then tempting to try to identify $\frL^{(c)}=\frL^{(1)\circ}\oplus \frL^{(1)\rm lc}\oplus \frL^{(2)\rm lc}$, but generally, the quasi-isomorphisms to $\frL^{(1,2)}$ are not homotopy transfers.
        \end{remark}
        
        \pagebreak
        
        \paragraph{Spans of $L_\infty$-algebras.} 
        We have the following result.
        \begin{theorem}\label{thm:1}
            Consider a pair of quasi-isomorphic $L_\infty$-algebras $\frL^{(1)}$ and $\frL^{(2)}$. Then there is a \uline{span of $L_\infty$-algebras}, i.e.~a third $L_\infty$-algebra $\frL^{(c)}$ together with quasi-isomorphisms $\sfp^{(1,2)}$ and $\sfe^{(1,2)}$ that fit into the diagram
            \begin{equation}
                \begin{tikzcd}
                    & \frL^{(c)}\arrow[dl,"\sfp^{(1)}",swap,shift right] \arrow[dr,"\sfp^{(2)}",shift left]
                    \\
                    \frL^{(1)} \arrow[ur,"\sfe^{(1)}",swap,shift right] & & \frL^{(2)}\arrow[ul,"\sfe^{(2)}",shift left]
                \end{tikzcd}
            \end{equation}
            such that the higher products on $\frL^{(1,2)}$ and $\sfe^{(1,2)}$ are obtained from a homotopy transfer and given by formulas~\eqref{eq:transfer_formulas}. We call $\frL^{(c)}$ the \uline{correspondence $L_\infty$-algebra}.
        \end{theorem}
        
        We perform the proof in a number of steps. Firstly, by definition, $\frL^{(1,2)}$ have isomorphic minimal models $\frL^{(1,2)\circ}$, and we can choose them to be identical: $\frL^{\circ}=\frL^{(1)\circ}=\frL^{(2)\circ}$. We then have the following diagram.
        \begin{equation}\label{eq:quasiIsomorphismThroughMinimalModel}    
            \begin{tikzcd}
                & 
                \frL^{(2)}\arrow[d,shift left]{}{\sfp^{(2)}}
                \\
                \frL^{(1)}\arrow[r,shift left]{}{\sfp^{(1)}}
                & 
                \frL^\circ\arrow[l,shift left]{}{\sfe^{(1)}}\arrow[u,shift left]{}{\sfe^{(2)}}
            \end{tikzcd}
        \end{equation}
        
        At this point, it turns out convenient to switch to the dual, Chevalley--Eilenberg picture, because we can phrase arguments in a way familiar from the BV formalism. Recall that any $L_\infty$-algebra $\frL$ is dual to a semi-free differential graded commutative algebra, called its Chevalley--Eilenberg algebra 
        \begin{equation}
            \sfCE(\frL)\ \coloneqq\ \left(\bigodot{}^\bullet(\frL[1])^*,Q\right)~,
        \end{equation}
        where $\odot^\bullet V$ denotes the symmetric tensor algebra over a graded vector space $V$, and $[k]$ denotes the shift-isomorphism 
        \begin{equation}
            [k]\,:\,V\ =\ \bigoplus_{i\in\IZ}V_i\ \rightarrow\ V[k]\ =\ \bigoplus_{i\in\IZ}(V[k])_i
            \ewith
            (V[k])_i\ \coloneqq\ V_{i+k}~. 
        \end{equation}
        However, all of the arguments we give below can be dualised to the perhaps less familiar picture of codifferential coalgebras. In particular, the differential $Q$ on $\sfCE(\frL)$ is the dual of the natural codifferential $D\coloneqq\mu_1[1]+\mu_2[1]+\ldots$ on the codifferential coalgebra $\bigodot^\bullet\frL[1]$. Therefore, the fact that the dualisation only exists in certain situations, e.g.~for degree-wise finite-dimensional $L_\infty$-algebras, is not a problem, and our formulation of our arguments in terms of Chevalley--Eilenberg algebras is indeed just for presentations sake. Also, since we only use this technology in this proof, we refrain from giving more details on Chevalley--Eilenberg algebras; for a detailed review in our conventions, see~e.g.~\cite{Jurco:2018sby}. 
        
        From this perspective, diagram~\eqref{eq:quasiIsomorphismThroughMinimalModel} translates to 
        \begin{equation}           
            \begin{tikzcd}
                \sfCE(\frL^\circ)\arrow[d,shift left]{}{\sfp^{(1)*}}\arrow[r,shift left]{}{\sfp^{(2)*}}
                & 
                \sfCE(\frL^{(2)})\arrow[l,shift left]{}{\sfe^{(2)*}}
                \\
                \sfCE(\frL^{(1)})\arrow[u,shift left]{}{\sfe^{(1)*}}
                & 
            \end{tikzcd}
        \end{equation}
        and we need to construct the push-out $\frA$, which is given by 
        \begin{equation}
            \frA\ \coloneqq\ \sfCE(\frL^{(1)}) \otimes_{\sfCE(\frL^\circ)} \sfCE(\frL^{(2)})\ =\ (\sfCE(\frL^{(1)}) \otimes \sfCE(\frL^{(2)}))/\caI
        \end{equation}
        with $\caI$ the ideal generated by expressions of the form
        \begin{equation}
            a(\sfp^{(1)*}(b))\otimes c-a\otimes(\sfp^{(2)*}(b)) c
        \end{equation}
        for all $a\in \sfCE(\frL^{(1)})$, $c\in\sfCE(\frL^\circ)$, and $b\in\sfCE(\frL^{(2)})$. Note that the ideal $\caI$ is a differential ideal, and the differential on $\frA$ is simply
        \begin{equation}
            \hat Q(a\otimes b)\ \coloneqq\ Q^{(1)}a\otimes b+(-1)^{|a|}a\otimes Q^{(2)}b~,
        \end{equation}
        where $Q^{(1,2)}$ are the differentials on $\sfCE(\frL^{(1,2)})$.
        
        The algebra $\frA$ is not yet the Chevalley--Eilenberg algebra of an $L_\infty$-algebra, because it is not semi-free\footnote{i.e.~free as a graded algebra (without differential)}. To remedy this, we use a Koszul--Tate-type complex quasi-isomorphic to $\frA$, very analogously to the BV formalism, see e.g.~\cite{Henneaux:1992}.\footnote{Recall that in the BV formalism, observables are defined as functions on field space modulo the ideal generated by the equations of motion. The resolution involves the introduction of anti-fields and the continuation of the BRST differential to these.}
        
        To this end, we introduce the graded commutative algebra 
        \begin{equation}
            \hat\frA\ \coloneqq\ \bigodot{}^\bullet(\frL^{\circ}\oplus\frL^{(1)}[1]\oplus\frL^{(2)}[1])^*~.
        \end{equation}
        There is a natural algebra homomorphism
        \begin{equation}
            \begin{aligned}
                g^*\,:\,\bigodot{}^\bullet(\frL^\circ\oplus \frL^\circ[1])^*\ &\rightarrow\ \hat\frA~,
                \\
                g^*|_{\odot^\bullet(\frL^{\circ}[1])^*}\ &\coloneqq\ \sfp^{(1)*}-\sfp^{(2)*}~, 
                \\
                g^*|_{\odot^\bullet(\frL^\circ)^*}\ &\coloneqq\ \sfi~,
            \end{aligned}
        \end{equation}
        where $\sfi:\bigodot{}^\bullet(\frL^\circ)^*\rightarrow \hat \frA$ is the evident inclusion. The algebra $\hat \frA$ becomes differential graded by virtue of the following result.
        \begin{proposition} 
            Consider the algebra $\hat\frA$ freely generated by $\xi^\alpha\in(\frL^{(1)}[1]\oplus\frL^{(2)}[1])^*$ and $\beta^i\in (\frL^{\circ})^*$. There is a differential $Q^{(c)}$ on $\hat \frA$ defined as 
            \begin{equation}\label{eq:thedifferential}
                \begin{aligned}
                    Q^{(c)}\xi^\alpha&\ \coloneqq\ (Q^{(1)}+Q^{(2)})\xi^\alpha,
                    \\
                    Q^{(c)}\beta^i&\ \coloneqq\ g^*\left(\sfs\beta^i+\sum_{k=1}^\infty\frac{1}{k!}\beta^{i_1}\cdots\beta^{i_k}P^i_{i_1\cdots i_k}\right)~,
                \end{aligned}
            \end{equation}
            where $P^i_{i_1\cdots i_n}$ are power series (without constant terms) in the $\sfs\beta^i$, and $\sfs$ is the shift isomorphism $\sfs=[-1]:(\frL^{\circ})^*\rightarrow (\frL^{\circ}[1])^*$.
        \end{proposition}
        \begin{proof}
            We have to show the existence of suitable $P^i_{i_1\cdots i_n}$ such that $Q^{(c)}$ squares to zero, which directly reduces to $(Q^{(c)})^2\beta^i=0$. We compute
            \begin{equation}
                \begin{aligned}
                    (Q^{(c)})^2\beta^i\ &=\ Q^{(c)}g^*\left(\sfs\beta^i+\sum_{k=1}^\infty\frac{1}{k!}\beta^{i_1}\cdots\beta^{i_k}P^i_{i_1\cdots i_k}\right)
                    \\
                    &=\ g^*(Q^\circ\sfs\beta^i)+\sum_{k=1}^\infty \frac{1}{(k-1)!} (Q^{(c)}\beta^{i_1})\beta^{i_2}\cdots\beta^{i_k}g^*(P^i_{i_1\cdots i_k})
                    \\
                    &\kern1cm+\sum_{k=1}^\infty\frac{(-1)^{|\beta^{i_1}|+\cdots+|\beta^{i_k}|}}{k!}\beta^{i_1}\ldots\beta^{i_k}g^*(Q^\circ P^i_{i_1\cdots i_k})~.
                \end{aligned}
            \end{equation}
            A construction for the $P^i_{i_1\cdots i_k}$ can be produced order by order in the\footnote{This is the common approach for demonstrating e.g.~the existence of the BV differential.} $\beta^i$. We start with the linear terms 
            \begin{equation}\label{eq:lowest_term}
                \beta^{i_1}P^i_{i_1}\ =\ -\sfs^{-1}\kappa(Q^\circ\sfs\beta^i)~,
            \end{equation}
            where $\kappa$ is defined as the linear extension of
            \begin{equation}
                \begin{aligned}
                    \kappa\,:\,\bigodot{}^\bullet(\frL^\circ[1])^*\ &\rightarrow\ \bigodot{}^\bullet(\frL^\circ[1])^*~,
                    \\
                    \sfs\beta^{i_1}\cdots\sfs\beta^{i_n}\ &\mapsto\ \frac{1}{n}\sfs\beta^{i_1}\cdots\sfs\beta^{i_n}~,
                \end{aligned}
            \end{equation}
            and $\sfs^{-1}$ is the inverse of $\sfs$, continued as a derivation to products. The resulting $Q^{(c)}$ satisfies
            \begin{equation}
                (Q^{(c)}_1)^2\beta^i\ =\ \caO(\beta)~,
            \end{equation}
            where we defined
            \begin{equation}
                Q^{(c)}_n\beta^i\ \coloneqq\ g^*\left(\sfs\beta^i+\sum_{k=1}^n\frac{1}{k!}\beta^{i_1}\cdots\beta^{i_k}P^i_{i_1\cdots i_k}\right)~.
            \end{equation}
            We can continue inductively. Say, we found the $P^i_{i_1\ldots i_k}$ up to order $n$, so that
            \begin{equation}
                Q^{(c)}_ng^*\left(\sfs\beta^i+\sum_{k=1}^n\frac{1}{k!}\beta^{i_1}\ldots\beta^{i_k}P^i_{i_1\cdots i_k}\right)\ =\ \caO(\beta^{n})~.
            \end{equation}
            Then we can choose the $P^i_{i_1\cdots i_{n+1}}$ so that 
            \begin{equation}
                \begin{aligned}
                    (Q^{(c)}_{n+1})^2\beta^i\ &\sim\ \frac{(-1)^{|\beta^{i_1}|+\ldots+|\beta^{i_n}|}}{n!}\beta^{i_1}\cdots\beta^{i_{n}}g^*(\sfs\beta^{i_{n+1}})g^*(P^i_{i_1\cdots i_{n+1}})
                    \\
                    &\hspace{1cm}+\frac{(-1)^{|\beta^{i_1}|+\cdots+|\beta^{i_n}|}}{n!}\beta^{i_1}\cdots\beta^{i_n}g^*(Q^\circ P^i_{i_1\cdots i_n})~,
                \end{aligned}
            \end{equation}
            where we dropped terms that are of order different than $n$ in the $\beta^i$. At the next order, we therefore need to solve
            \begin{equation}
                \beta^{i_1}\cdots\beta^{i_{n}}\,\sfs \beta^{i_{n+1}}P^i_{i_1\cdots i_{n+1}}\ =\ \beta^{i_1}\cdots\beta^{i_{n}}Q^\circ P^i_{i_1\cdots i_n}~,
            \end{equation}
            which always has a solution. A particular solution is given by 
            \begin{equation}
                \beta^{i_{n+1}}P^i_{i_1\cdots i_{n+1}}\ =\ \sfs^{-1} \kappa(Q^\circ P^i_{i_1\cdots i_n})~,
            \end{equation}        
            as is readily seen by applying $\sfs$ to both sides and multiplying the results by $\beta^{i_1}\cdots\beta^{i_{n}}$.
        \end{proof}
        
        Having defined a semi-free differential graded commutative algebra, we can convince ourselves that it is of the right size and that there are homotopy transfers from the $L_\infty$-algebra defined by $\hat\frA$ to either $\frL^{(1,2)}$. 
        \begin{proposition}
            The differential graded commutative algebra $(\hat \frA, Q^{(c)})$ is the Chevalley--Eilenberg algebra of an $L_\infty$-algebra $\frL^{(c)}$ quasi-isomorphic to both $\frL^{(1)}$ and $\frL^{(2)}$. In fact, there are homotopy transfers from $\frL^{(c)}$ to $\frL^{(1)}$ and $\frL^{(2)}$.
        \end{proposition}
        \begin{proof}
            As a vector space, the $L_\infty$-algebra $\frL^{(c)}$ is given by 
            \begin{equation}
                \frL^{(c)}\ =\ \frL^{(1)}\oplus\frL^\circ[-1]\oplus\frL^{(2)}~.
            \end{equation}
            To extract the differential, we need to consider the linear part of $Q^{(c)}$. Because $Q^\circ$ does not have a linear term, equation~\eqref{eq:lowest_term} implies that the linear differential $Q^{(c)}_\text{lin}$ acts according to 
            \begin{equation}
                \begin{aligned}
                    Q^{(c)}_\text{lin}\xi^\alpha\ &\coloneqq\ (Q^{(1)}_\text{lin}+Q^{(2)}_\text{lin})\xi^\alpha,
                    \\
                    Q^{(c)}_\text{lin}\beta^i\ &\coloneqq\ g^*_\text{lin}(\sfs\beta^i)~.
                \end{aligned}
            \end{equation}
            Hence, the cochain complex underlying $\frL^{(c)}$ takes the form
            \begin{equation}\label{11}
                \sfCh(\frL^{(c)})\ =\ 
                \left( 
                \begin{tikzcd}
                    \ldots \frL^{(1)}_0 \arrow[r,"\mu_1^{(1)}"]  \arrow[rd,"\sfp^{(1)}_\text{lin}"]  
                    & 
                    \frL^{(1)}_1 \arrow[r,"\mu_1^{(1)}"] \arrow[rd,"\sfp^{(1)}_\text{lin}"]      
                    &
                    \frL^{(1)}_2 \arrow[r,"\mu_1^{(1)}"] \arrow[rd,"\sfp^{(1)}_\text{lin}"]
                    &
                    \frL^{(1)}_3 \ldots
                    \\
                    \frL^\circ_{-1} \arrow[r,"0"] 
                    & 
                    \frL^\circ_0 \arrow[r,"0"] 
                    &
                    \frL^\circ_1 \arrow[r,"0"] 
                    &
                    \frL^\circ_2
                    \\
                    \ldots  \frL^{(2)}_0 \arrow[r,"\mu_1^{(2)}"']  \arrow[ru,"-\sfp^{(2)}_\text{lin}"']  
                    & 
                    \frL^{(2)}_1 \arrow[r,"\mu_1^{(2)}"'] \arrow[ru,"-\sfp^{(2)}_\text{lin}"']      
                    &
                    \frL^{(2)}_2 \arrow[r,"\mu_1^{(2)}"'] \arrow[ru,"-\sfp^{(2)}_\text{lin}"']
                    &
                    \frL^{(2)}_3 \ldots
                \end{tikzcd}\right)~.
            \end{equation}
            The cohomology of $\frL^{(1)}\oplus \frL^{(2)}$ is $\frL^\circ\oplus \frL^\circ$, but in $\frL^{(c)}$, only a subspace isomorphic to $\frL^\circ$ is in the kernel of the differential $\mu_1$. Therefore, it is clear that the cohomology of $\frL^{(c)}$ is isomorphic to $\frL^{\circ}$. 
            
            In order to show that there are homotopy transfers from $\frL^{(c)}$ to both $\frL^{(1)}$ and $\frL^{(2)}$, we can apply \cref{cor:projections_are_homotopy_transfer} because the obvious projections from $\frL^{(c)}$ to $\frL^{(1)}$ and $\frL^{(2)}$ are strict maps of $L_\infty$-algebras. 
            
            Let us illustrate the details. Since the situation is symmetric, we can focus on $\frL^{(1)}$. Let $\sfp$ be the obvious projection from $\frL^{(c)}$ to $\frL^{(1)}$. We note that this is a strict morphism of $L_\infty$-algebras, which can be seen from \eqref{eq:thedifferential}. Take an element $\xi^\alpha\in (\frL^{(1)}[1])^*$, by the first line of \eqref{eq:thedifferential}, we have $Q^{(c)}\sfp^*(\xi^\alpha)=\sfp^*(Q^{(1)}\xi^\alpha)$, where $\sfp^*\,:\,\bigodot{}^\bullet(\frL^{(1)}[1])^*\to \bigodot{}^\bullet(\frL^{(1)}[1] \oplus \frL^\circ\oplus \frL^{(2)}[1]  )^*$ is the obvious inclusion. In the $L_\infty$-algebra language, this translates to the equation $\sfp\circ \mu^{(c)}_k= \mu^{(1)}_k\circ (\sfp^{\otimes k})$. Next, consider the deformation retracts
            \begin{equation}
                \begin{tikzcd}
                    \ar[loop,out=160,in=190,distance=21,"\sfh^{(1,2)}" left] (\frL^{(1,2)},\mu_1^{(1,2)})\arrow[r,shift left]{}{\sfp^{(1,2)}} & (\frL^{\circ},0)\arrow[l,shift left]{}{\sfe^{(1,2)}}
                \end{tikzcd}
            \end{equation}
            satisfying the side conditions \eqref{eq:HT_side_conditions}.
            We define an embedding $\sfe\,:\,\frL^{(1)}\rightarrow \frL^{(c)}$ and a contracting homotopy $\sfh\,:\, \frL^{(c)}\rightarrow \frL^{(c)}$ by
            \begin{equation}
				\begin{aligned}
                    \sfe(b^{(1)})\ &\coloneqq\ (b^{(1)},0,(\sfe^{(2)}\circ\sfp^{(1)})(b^{(1)}))~,
                    \\
                    \sfh(b^{(1)},b^\circ,b^{(2)})\ &\coloneqq\ \big(0,-\sfe^{(2)}(b^\circ),\sfh^{(2)}(b^{(2)})\big)
				\end{aligned}
			\end{equation}
            for all $\sfb^{(1)}\in \frL^{(1)}$, $\sfb^\circ\in \frL^\circ[-1]$, and $\sfb^{(2)}\in \frL^{(2)}$. One can check that $\sfh$ satisfies the side conditions. We thus arrive at the special deformation retract 
            \begin{equation}
                \begin{tikzcd}
                    \ar[loop,out=160,in=190,distance=21,"\sfh" left] (\frL^{(c)},\mu_1^{(c)})\arrow[r,shift left]{}{\sfp} & (\frL^{(1)},\mu_1^{(1)})\arrow[l,shift left]{}{\sfe}
                \end{tikzcd}.
            \end{equation}
            
            Because $\sfp\colon \frL^{(c)}\to  \frL^{(1)}$ is a strict map of $L_\infty$-algebras, we can apply \cref{cor:projections_are_homotopy_transfer}: we have $\sfp\circ\sfh=0$ and so, $\sfp\circ\sfE_j=0$ for all $j>1$, where $\sfE_j$ is defined in~\eqref{eq:QuasiIsomorphismInTransfer}, and the higher products on $\frL^{(1)}$ induced by the homotopy transfer and given by~\eqref{eq:TransferredHigherProducts} coincide with the original higher products of $\frL^{(1)}$.
        \end{proof}
        
        With the last lemma, the proof of \cref{thm:1} is complete. 
        
        \paragraph{Lifting of homotopy transfers.} As a trivial example, consider the lift of a quasi-isomorphism arising from a homotopy transfer 
        \begin{equation}
            \begin{tikzcd}
                \ar[loop,out=160,in=200,distance=20,"\sfh^{(1)}" left] \big(\frL^{(1)},\mu_1^{(1)}\big)\arrow[r,shift left]{}{\sfp^{(21)}} & \big(\frL^{(2)},\mu^{(2)}_1\big)\arrow[l,shift left]{}{\sfe^{(12)}}
            \end{tikzcd}
        \end{equation}
        Such a quasi-isomorphism trivially lifts to the span
        \begin{equation}
            \begin{tikzcd}
                & \frL^{(1)}\arrow[dl,"\sfid",swap,shift right] \arrow[dr,"\sfp^{(21)}",shift left]
                \\
                \frL^{(1)} \arrow[ur,"\sfid",swap,shift right] & & \frL^{(2)} \arrow[ul,"\sfe^{(12)}",shift left]
            \end{tikzcd}
        \end{equation}
        More interesting examples, in particular with regards to applications in quantum field theory, are presented in \cref{sec:toy_model,sec:PCM,sec:twistors}.
        
        \paragraph{Composition of spans of $L_\infty$-algebras.} Given a pair of quasi-isomorphic $L_\infty$-algebras, it is clear that, generically, there are several spans of $L_\infty$-algebras between them. Their correspondence $L_\infty$-algebras are related by a quasi-isomorphisms, which then also relate the various projection and embedding maps in an evident manner. 
        
        This ambiguity also induces an ambiguity in the composition of spans, which can simply be defined as spans between the correspondence $L_\infty$-algebras. Therefore, $L_\infty$-algebras as objects with spans of $L_\infty$-algebras as morphisms do not form a category or groupoid (because the morphisms are evidently invertible) but can be regarded as a quasi-groupoid. Alternatively, we can simply quotient by this ambiguity and regard different spans between the same pair of $L_\infty$-algebra as equivalent, rendering composition again unique and associative. Since our interest in spans is mostly due to computational advantages, this distinction is largely irrelevant in the following.
        
        \paragraph{Quasi-isomorphisms of cyclic $L_\infty$-algebras.} Before coming to the physical applications of spans of $L_\infty$-algebras within field theories, let us briefly explore the mathematical uses. Recall that the definition of morphisms between cyclic $L_\infty$-algebras is somewhat problematic. Essentially, the reason for the encountered problems is the fact that a cyclic structure corresponds to a symplectic structure on the underlying graded vector space, and hence, one is looking for a reasonable category of symplectic manifolds. In particular, we have seen in that the condition~\eqref{eq:compatibility_metric_structure} only works for morphisms $\phi$ for which the cochain map $\phi_1$ an injection. Our perspective solves this problem at least for quasi-isomorphisms, and we can make the following definition.
        \begin{definition}
            Given two cyclic $L_\infty$-algebras $(\frL^{(1,2)},\langle-,-\rangle^{(1,2)})$, a \uline{metric} or \uline{cyclic quasi-isomorphism} is a third cyclic $L_\infty$-algebra $(\frL^{(c)},\langle-,-\rangle^{(c)})$ such that we have a span of $L_\infty$-algebras
            \begin{subequations}
                \begin{equation}
                    \begin{tikzcd}
                        & \frL^{(c)}\arrow[dl,"\sfp^{(1)}",swap,shift right] \arrow[dr,"\sfp^{(2)}",shift left]
                        \\
                        \frL^{(1)} \arrow[ur,"\sfe^{(1)}",swap,shift right] & & \frL^{(2)}\arrow[ul,"\sfe^{(2)}",shift left]
                    \end{tikzcd}
                \end{equation}
                and
                \begin{equation}\label{eq:condition_cyclic_span}
                    \langle a_1,a_2\rangle^{(1)}\ =\ \langle \sfe^{(1)}_1(a_1),\sfe^{(1)}_1(a_2)\rangle^{(c)}
                    \eand 
                    \langle b_1,b_2\rangle^{(2)}\ =\ \langle \sfe^{(2)}_1(b_1),\sfe^{(2)}_1(b_2)\rangle^{(c)}
                \end{equation}
                for all $a_{1,2}\in \frL^{(1)}$ and $b_{1,2}\in \frL^{(2)}$.
            \end{subequations}        
        \end{definition}        
        \begin{remark}
            We note that the full condition~\eqref{eq:compatibility_metric_structure} for the injective quasi-isomorphisms $\sfe^{(1,2)}$ is automatically satisfied if~\eqref{eq:condition_cyclic_span} holds. Consider e.g.~the Hodge-type decomposition
            \begin{equation}
                \frL^{(c)}\ =\ H\oplus B \oplus C
            \end{equation}
            with $H=\ker(\mu_1^{(1)}\sfh^{(1)}_1+\sfh^{(1)}_1\mu_1^{(1)})$, $B=\im(\mu_1^{(1)})$, and $C=\im(\sfh^{(1)})$ of the initial special deformation retract between $\frL^{(c)}$ and $\frL^{(1)}$ before deformation. The metric structure on $\frL^{(c)}$ then is necessarily of the block matrix form
            \begin{equation}
                \langle c_1,c_2\rangle^{(c)}\ =\ c_1^T\begin{pmatrix} \omega_{HH} & 0 & 0 \\ 0 & 0 & \omega_{BC} \\ 0 & \omega_{CB} & 0 \end{pmatrix} c_2
            \end{equation}
            for some block matrices $\omega_{HH}$, $\omega_{BC}$, and $\omega_{CB}$. The formulas for the maps $\sfe^{(1)}_k$ given in~\eqref{eq:QuasiIsomorphismInTransfer} then show that we automatically have 
            \begin{equation}
                \sum_{\substack{j+k=i\\j,k\geq 1}}\inner{\sfe^{(1)}_j(a_1,\ldots,a_j)}{\sfe^{(1)}_k(a_{j+1},\ldots,a_{j+k}))}^{(2)}\ =\ 0~,
            \end{equation}
            the additional condition in~\eqref{eq:compatibility_metric_structure}.
        \end{remark}
        
        \subsection{Application to perturbative quantum field theory}
        
        In this section, we provide a very concise review of the dictionary that translates between perturbative field theories and $L_\infty$-algebras, as well as the implications for our above results. For further details, see e.g.~\cite{Kajiura:2003ax,Doubek:2017naz,Hohm:2017pnh,Jurco:2018sby,Jurco:2020yyu,Borsten:2021hua}. 
        
        \paragraph{Observables in a classical field theory.}
        The kinematical data of a perturbative classical field theory is the field space $\frF$, a vector space or module usually consisting of the sections of some vector bundle. The dynamics of the theory are governed by the equations of motion, which are the stationary points of an action functional $S$ on the field space. Note that the equations of motion generate the ideal $\caI$ in the ring of functions on field space which is generated by functions on $\frF$ vanishing for classical solutions.
        
        This description may contain redundancies known as gauge symmetries, i.e.~an action of a group (of gauge transformations\footnote{not to be confused with the gauge group}) $\caG\curvearrowright\frF$ such that the true kinematical data is given by the orbit space $\frF/\caG$. This evidently requires the action functional and the equations of motion to be invariant and covariant, respectively, under the group action. 
        
        The classical observables of a field theory are then identified with the quotient ring $\caR/\caI$, where $\caR$ is the ring of functions on the orbit space $\frF/\caG$.
        
        \paragraph{Batalin--Vilkovisky formalism.}
        Quotient spaces are notoriously inconvenient to work with, and a useful alternative is provided by the Batalin--Vilkovisky (BV) complex~\cite{Batalin:1977pb,Batalin:1981jr,Batalin:1984jr,Batalin:1984ss,Batalin:1985qj,Schwarz:1992nx}. In this description, both the quotients by the group of gauge transformations $\caG$ and the ideal $\caI$ are replaced by a cochain complex whose cohomology is the actual quotient. Gauge transformations are dealt with the Chevalley--Eilenberg algebra of the corresponding gauge algebroid, introducing ghost fields. The equations of motion are divided out by introducing additional anti-fields, leading to a Koszul--Tate complex.\footnote{We note that the original motivation for implementing the BV formalism stemmed from the perturbative description of quantum gauge theories, but the Koszul--Tate part exists also for theories without gauge symmetry.} The resulting BV complex is a cochain complex with the differential encoding the equations of motion and the gauge transformations. The cochains form, in fact, a differential graded commutative algebra, which, as we mentioned above, is the dual of an $L_\infty$-algebra.
        
        \paragraph{Direct correspondence: Maurer--Cartan action.}
        We can also establish the link between a field theory and an $L_\infty$-algebra more directly. The data of any perturbative field theory\footnote{read: with a field space given by sections of a vector bundle} can be cast in the form of a cyclic $L_\infty$-algebra $\frL$. Here, $\frL_1$ is the vector space or module of fields, while $\frL_2$ is the space of anti-fields. There is a metric structure of degree~$-3$, providing a non-degenerate pairing between $\frL_2$ and $\frL_1^*$. If gauge symmetries (respectively, higher gauge symmetries) are present, we also have a non-trivial subspace $\frL_0$ (respectively, $\frL_0$, $\frL_{-1}$, etc.), the space of (infinitesimal) gauge parameters or ghosts (respectively, ghosts, ghost-of-ghosts, etc.), as well as $\frL_3$ (respectively, $\frL_3$, $\frL_4$, etc.), the space of anti-ghosts (respectively, anti-ghosts, anti-ghosts-of-ghosts, etc.).
        
        The cyclic higher products on $\frL_1$ are then uniquely defined by identifying the field theory's classical action with the homotopy Maurer--Cartan action of $\frL$,
        \begin{equation}
            S\ \stackrel{!}{=}\ \sum_{i\geq 0}\frac{1}{(i+1)!}\inner{a}{\mu_i(a,\ldots,a)}~,
        \end{equation}
        where $a\in\frL_1$. The remaining higher products for all elements in $\frL$ are defined either via the gauge transformations of the fields or by writing the BV action in a particular way, cf.~\cite{Jurco:2018sby} (see also~\cite{Cattaneo:0010172}).
        
        Altogether, perturbative field theories with action principles are in one-to-one correspondence with $L_\infty$-algebras endowed with a metric structure of degree~$-3$.
        
        \paragraph{Perturbation theory.}
        Interestingly, the way that tree-level perturbation theory is usually described within quantum field theory directly translates to the homological perturbation lemma. In order to compute a tree-level amplitude, we draw all relevant Feynman diagrams, amputate the external legs by the Lehmann--Symanzik--Zimmermann (LSZ) reduction formula, replacing them with labels of gauge-fixed free fields. The construction of the minimal model via the special deformation retract~\eqref{eq:MM_deformation_retract} and the formulas~\eqref{eq:transfer_formulas} proceeds exactly in the same manner. The cohomology $H^\bullet_{\mu_1}(\frL)$ describes the free fields up to gauge transformations, and the recursion relation for the $\sfE_i$ produces the tree-level Feynman diagram expansion with the $n$-point tree-level scattering amplitudes themselves are identified with the expressions
        \begin{equation}\label{eq:tree_scattering_amplitude}
            \caA(\phi_1,\ldots,\phi_n)\ =\ \tfrac{1}{n!}\inner{\phi_1}{\mu_{n-1}^\circ(\phi_2,\ldots,\phi_n)}^\circ
        \end{equation}
        for $\phi_1,\ldots,\phi_n$ elements of the minimal model $\frL^\circ$. This observation has been made many times, see e.g.~\cite{Kajiura:2001ng,Kajiura:2003ax} in the context of string field theory and~\cite{Nutzi:2018vkl,Macrelli:2019afx,Arvanitakis:2019ald} in the context of field theories.
        
        \paragraph{Semi-classical equivalence.}
        A suitable definition of equivalence between two classical field theories must start from the question of which properties equivalent quantum field theories have to share. An isomorphic solution space is clearly not enough, as e.g.~theories of a single massless scalar field on Minkowski space $\IR^{1,n}$ with canonical kinematic term and arbitrary polynomial potentials all have isomorphic solution spaces, parametrised e.g.~by boundary data on some Cauchy surface. A good quantity to preserve is certainly the tree-level S-matrix containing the scattering amplitudes~\eqref{eq:tree_scattering_amplitude}, as we think of these as determining all measurable quantities. We note that this notion of equivalence called \uline{semi-classical equivalence} covers the standard operations of integrating fields in and out.
        
        This form of equivalence is also mathematically preferable, as two field theories have S-matrices related by a similarity transformation if and only if their corresponding $L_\infty$-algebras have isomorphic minimal models and are thus quasi-isomorphic.
        
        \paragraph{Correspondences of $L_\infty$-algebras.}
        As mentioned above, semi-classically equivalent field theories have quasi-isomorphic $L_\infty$-algebras, and as a consequence, isomorphic minimal models, which can be computed via homotopy transfer. In many situations, the two $L_\infty$-algebras are directly linked by a homotopy transfer, e.g.~when they are related by integrating out fields. In some situations, however, this is not the case, and physicists are already implicitly working with spans to describe these. In the following sections, we shall discuss three examples in detail: a simple example based on different blow ups of scalar field theory, (non-Abelian) T-duality in the case of the principal chiral model, and the Penrose--Ward transform.
        
        An important remark regarding the application to field theory is the following. The correspondence $L_\infty$-algebra we constructed in~\cref{thm:1} using the minimal model is usually inconvenient from the field theoretic perspective, as it splits a field into its on-shell and off-shell components. Most of the time, we are interested in a correspondence $L_\infty$-algebra with all elements true, unrestricted fields on a given space-time. This makes constructing a `good' correspondence $L_\infty$-algebra a bit more non-trivial. However, completing on-shell fields in a minimal model by trivial pairs to true, unrestricted fields on a given space-time is mostly straightforward\footnote{up to analytical difficulties, but these can be removed by restricting to functions bounded at infinity}. In concrete computations with field theories, one therefore replaces the summand $\frL^\circ[-1]$ in $\frL^{(c)}$ with an enlarged copy, $\hat \frL^\circ[-1]$, which contains these additional trivial pairs. Again, there are no fundamental obstructions to the existence of physically `good' correspondence $L_\infty$-algebras.
        
        \paragraph{Quantum level.}
        In this paper, we shall exclusively work at the tree level. Still, let us briefly comment on the extension to the quantum picture. In the BV formalism, the differential in the BV complex is deformed by a second order differential operator, and this deformation produces a differential graded algebra that is dual not to an $L_\infty$-algebra, but to a loop $L_\infty$-algebra, as defined in~\cite{Zwiebach:1992ie,Markl:1997bj}. Loop $L_\infty$-algebras, however, are still homotopy algebras, and come with a homotopy transfer of their structure to a minimal model, cf.~\cite{Doubek:2017naz,Jurco:2019yfd}, which now governs the quantum scattering amplitudes. Two field theories are then quantum equivalent, if their loop $L_\infty$-algebras have the same minimal model. Note that all issues regarding regularisation and renormalisation have been ignored in this discussion. For a rigorous treatment of these, see~\cite{Costello:2011aa,Costello:2016vjw,Costello:2021jvx}. 
        
        At the quantum level, an additional benefit of lifting quasi-isomorphisms to spans of $L_\infty$-algebras is that they make equivalences more evident. As a homotopy transfer essentially amounts to integrating in/out some fields, one can identify those integrations that are compatible with both homotopy algebra and loop homotopy algebra structures. That is, one can check whether the field redefinitions included in the homotopy transfer induce Jacobians when applied to the path integral measure. For more details on this as wells as the notion of equivalence of perturbative quantum field theories at quantum level, see also the discussion in~\cite{Borsten:2021gyl}.
        
        \section{Blowing up vertices in scalar field theory}\label{sec:toy_model}
        
        Let us start with a simple example and consider scalar field theory with a sextic potential, i.e.~the action
        \begin{equation}\label{eq:phi6}
            S^{(b)}\ \coloneqq\ \int\rmd^dx\,\big\{\tfrac12\phi\wave\phi-\tfrac{\lambda^2}{6!}\phi^6\big\}
        \end{equation}
        for a single scalar field theory $\phi\in\Omega^0(\IM^d)$ on $d$-dimensional Minkowski space $\IM^d$. Here, $\wave$ is the d'Alembertian and $\lambda\in\IR$. In the following, we shall blow up the interaction vertex by introducing auxiliary fields in two different ways: in one theory, as two quartic vertices, and in another as a cubic and a quintic vertex. In order to go from one theory to the other, one has to both integrate in and out auxiliary fields, leading naturally to a span of $L_\infty$-algebras.
        
        \subsection{Homotopy algebraic formulation of the involved theories}
        
        \paragraph{Blow ups of interaction vertices.} 
        By introducing auxiliary fields $\chi_1,\chi_2\in\Omega^0(\IM^d)$, we can blow up the interaction vertex in~\eqref{eq:phi6} into two quartic vertices, 
        \begin{equation}
            S^{(1)}\ \coloneqq\ \int\rmd^dx\,\big\{\tfrac12\phi\wave\phi+\chi_1\chi_2-\tfrac{\lambda}{\sqrt{6!}}\phi^3\chi_1-\tfrac{\lambda}{\sqrt{6!}}\phi^3\chi_2\big\}\,.
        \end{equation}
        Alternatively, we can introduce two auxiliary fields $\psi_1,\psi_2\in\Omega^0(\IM^d)$ and blow up the interaction vertex in~\eqref{eq:phi6} into cubic and quintic vertices,
        \begin{equation}
            S^{(2)}\ \coloneqq\ \int\rmd^dx\,\big\{\tfrac12\phi\wave\phi+\psi_1\psi_2-\tfrac{\lambda}{\sqrt{6!}}\phi^2\psi_1-\tfrac{\lambda}{\sqrt{6!}}\phi^4\psi_2\big\}\,.
        \end{equation}
        Finally, there is the following blow-up using all the above auxiliary fields as well as $\xi_1,\xi_2\in\Omega^0(\IM^d)$,
        \begin{equation}
            S^{(c)}\ \coloneqq\ \int\rmd^dx\big\{\tfrac12\phi\wave\phi+\chi_1\chi_2+\psi_1\psi_2+\xi_1\xi_2-\tfrac{\lambda}{\sqrt{6!}}\phi^2\psi_1-\psi_2\phi\chi_1-\chi_2\phi\xi_1-\tfrac{\lambda}{\sqrt{6!}}\xi_2\phi^2\big\}\,.
        \end{equation}
        Note that this blow up is a `common refinement' of the previous two in the following sense: if we integrate out the fields $\psi_1,\psi_2$ and $\chi_1,\chi_2$ in $S^{(c)}$, we recover $S^{(1)}$; if we integrate out $\xi_1,\xi_2$ and $\psi_1,\psi_2$, on the other hand, we recover $S^{(2)}$.
        
        \paragraph{Homotopy algebra for $S^{(b)}$.} 
        The homotopy algebra $\frL^{(b)}$ corresponding to the action $S^{(b)}$ has the underlying cochain complex
        \begin{subequations}
            \begin{equation}
                \sfCh(\frL^{(b)})\ \coloneqq\ \left( 
                \begin{tikzcd}
                    \underbrace{\spacer{2ex}\stackrel{\phi}{\Omega^0(\IM^d)}}_{\eqqcolon\,\frL^{(b)}_1} \arrow[r,"\wave"] 
                    &
                    \underbrace{\spacer{2ex}\stackrel{\phi^+}{\Omega^0(\IM^d)}}_{\eqqcolon\,\frL^{(b)}_2}
                \end{tikzcd}\right)\,,
            \end{equation}
            together with the only non-trivial higher product
            \begin{equation}
                \mu^{(b)}_5(\phi,\ldots,\phi)\ \coloneqq\ -\lambda^2\phi^6
            \end{equation}
            for all $\phi\in\frL^{(b)}_1$; the general expression follows from polarisation. The metric structure is the usual integral
            \begin{equation}
                \inner{\phi}{\phi^+}_{\frL^{(b)}}\ \coloneqq\ \int\rmd^dx\,\phi\phi^+
            \end{equation}
            for all $\phi\in\frL^{(b)}_1$ and $\phi^+\in\frL^{(b)}_2$. Note that we regard the integral as formal expressions; to be precise, we would have to restrict our field space to $L^2$-functions, cf.~e.g.~\cite{Macrelli:2019afx} for a detailed discussion.
        \end{subequations}     
        
        \paragraph{Homotopy algebra for $S^{(1)}$.}
        The homotopy algebra $\frL^{(1)}$ corresponding to the action $S^{(1)}$ has underlying cochain complex
        \begin{subequations}\label{eq:L_infty_scal_1}
            \begin{equation}\label{eq:L_infty_scal_1a}
                \sfCh(\frL^{(1)})\ \coloneqq\ \left( 
                \begin{tikzcd}
                    \stackrel{\phi}{\Omega^0(\IM^d)}\arrow[r,"\wave"]  & \stackrel{\phi^+}{\Omega^0(\IM^d)}
                    \\[-0.5cm]
                    \stackrel{\chi_1}{\Omega^0(\IM^d)}\arrow[start anchor=south east, end anchor={150},dr,"\sfid",pos=0.95,swap] & 
                    \stackrel{\chi_1^+}{\Omega^0(\IM^d)}
                    \\[-0.5cm]
                    \underbrace{\spacer{2ex}\stackrel{\chi_2}{\Omega^0(\IM^d)}}_{\eqqcolon\,\frL^{(1)}_1}\arrow[start anchor={30}, end anchor= south west,crossing over,ur,"\sfid",pos=0.95]
                    &
                    \underbrace{\spacer{2ex}\stackrel{\chi^+_2}{\Omega^0(\IM^d)}}_{\eqqcolon\,\frL^{(1)}_2}
                \end{tikzcd}\right)\,,
            \end{equation}
            and the only non-trivial higher products are obtained by polarisation of 
            \begin{equation}\label{eq:mu_3_scal_1}
                \mu_3^{(1)}\left(
                \begin{pmatrix}
                    \phi & \phi^+
                    \\
                    \chi_1 & \chi_1^+
                    \\
                    \chi_2 & \chi_2^+
                \end{pmatrix},
                \begin{pmatrix}
                    \phi & \phi^+
                    \\
                    \chi_1 & \chi_1^+
                    \\
                    \chi_2 & \chi_2^+
                \end{pmatrix},
                \begin{pmatrix}
                    \phi & \phi^+
                    \\
                    \chi_1 & \chi_1^+
                    \\
                    \chi_2 & \chi_2^+
                \end{pmatrix}
                \right)\ =\ -\frac{3!\lambda}{\sqrt{6!}}
                \begin{pmatrix}
                    0 & 3\phi^2\chi_1+3\phi^2 \chi_2
                    \\
                    0 & \phi^3
                    \\
                    0 & \phi^3
                \end{pmatrix}
            \end{equation}
            for all $(\phi,\chi_1,\chi_2)\in\frL^{(2)}_1$ and $(\phi^+,\chi^+_1,\chi^+_2)\in\frL^{(2)}_2$, and where here and below the positions of the field components correspond to those in~\eqref{eq:L_infty_scal_1a}. The metric structure is the evident one, 
            \begin{equation}
                \stretchleftright[1000]{<}{
                    \begin{pmatrix}
                        \phi & \phi^+
                        \\
                        \chi_1 & \chi_1^+
                        \\
                        \chi_2 & \chi_2^+
                    \end{pmatrix},
                    \begin{pmatrix}
                        \tilde\phi & \tilde\phi^+
                        \\
                        \tilde\chi & \tilde\chi_1^+
                        \\
                        \tilde\chi & \tilde\chi_2^+ 
                    \end{pmatrix}
                }{>}\raisebox{-20pt}{${}_{\frL^{(1)}}$}\ \coloneqq\ \int\rmd^dx\,\big\{\phi\tilde\phi^++\chi_1\tilde\chi_1^++\chi_2\tilde\chi_2^++\tilde\phi\phi^++\tilde\chi_1\chi_1^++\tilde\chi_2\chi_2^+\big\}
            \end{equation}
            for all $(\phi,\chi_1,\chi_2)\in\frL^{(2)}_1$ and $(\phi^+,\chi^+_1,\chi^+_2)\in\frL^{(2)}_2$. 
        \end{subequations}     
        
        \paragraph{Homotopy algebra for $S^{(2)}$.}
        The homotopy algebra $\frL^{(2)}$ corresponding to the action $S^{(2)}$ has underlying cochain complex
        \begin{subequations}\label{eq:L_infty_scal_2}
            \begin{equation}\label{eq:L_infty_scal_2a}
                \sfCh(\frL^{(2)})\ \coloneqq\ \left( 
                \begin{tikzcd}
                    \stackrel{\phi}{\Omega^0(\IM^d)}\arrow[r,"\wave"]  & \stackrel{\phi^+}{\Omega^0(\IM^d)}
                    \\[-0.5cm]
                    \stackrel{\psi_1}{\Omega^0(\IM^d)}\arrow[start anchor=south east, end anchor={150},dr,"\sfid",pos=0.95,swap] & 
                    \stackrel{\psi_1^+}{\Omega^0(\IM^d)}
                    \\[-0.5cm]
                    \underbrace{\spacer{2ex}\stackrel{\psi_2}{\Omega^0(\IM^d)}}_{\eqqcolon\,\frL^{(2)}_1}\arrow[start anchor={30}, end anchor= south west,crossing over,ur,"\sfid",pos=0.95]
                    &
                    \underbrace{\spacer{2ex}\stackrel{\psi^+_2}{\Omega^0(\IM^d)}}_{\eqqcolon\,\frL^{(2)}_2}
                \end{tikzcd}\right)\,,
            \end{equation}
            and the only non-trivial higher products are obtained by polarisation of 
            \begin{equation}\label{eq:L_infty_scal_2_products}
                \begin{aligned}
                    \mu^{(2)}_2\left(
                    \begin{pmatrix} 
                        \phi & \phi^+
                        \\
                        \psi_1 & \psi_1^+
                        \\
                        \psi_2 & \psi_2^+
                    \end{pmatrix},
                    \begin{pmatrix} 
                        \phi & \phi^+
                        \\
                        \psi_1 & \psi_1^+
                        \\
                        \psi_2 & \psi_2^+
                    \end{pmatrix}
                    \right)\ &\coloneqq -\frac{2\lambda}{\sqrt{6!}}
                    \begin{pmatrix}
                        0 & 2\psi_1\phi
                        \\
                        0 & \phi^2
                        \\
                        0 & 0
                    \end{pmatrix},
                    \\
                    \mu^{(2)}_4\left(
                    \begin{pmatrix} 
                        \phi & \phi^+
                        \\
                        \psi_1 & \psi_1^+
                        \\
                        \psi_2 & \psi_2^+
                    \end{pmatrix},
                    \begin{pmatrix} 
                        \phi & \phi^+
                        \\
                        \psi_1 & \psi_1^+
                        \\
                        \psi_2 & \psi_2^+
                    \end{pmatrix},
                    \begin{pmatrix} 
                        \phi & \phi^+
                        \\
                        \psi_1 & \psi_1^+
                        \\
                        \psi_2 & \psi_2^+
                    \end{pmatrix},
                    \begin{pmatrix} 
                        \phi & \phi^+
                        \\
                        \psi_1 & \psi_1^+
                        \\
                        \psi_2 & \psi_2^+
                    \end{pmatrix}    
                    \right)\ &\coloneqq\ -\frac{4!\lambda}{\sqrt{6!}}
                    \begin{pmatrix}
                        0 & 4\phi^3 \psi_2
                        \\
                        0 & 0
                        \\
                        0 & \phi^4
                    \end{pmatrix}
                \end{aligned}
            \end{equation}
            for all $(\phi,\psi_1,\psi_2)\in\frL^{(2)}_1$ and $(\phi^+,\psi^+_1,\psi^+_2)\in\frL^{(2)}_2$ and where the positions of the component fields refer to diagram~\eqref{eq:L_infty_scal_2a}. The metric structure is again evident, 
            \begin{equation}
                \stretchleftright[1000]{<}{
                    \begin{pmatrix}
                        \phi & \phi^+
                        \\
                        \psi_1 & \psi_1^+
                        \\
                        \psi_2 & \psi_2^+
                    \end{pmatrix},
                    \begin{pmatrix}
                        \tilde\phi & \tilde\phi^+
                        \\
                        \tilde\psi_1 & \tilde\psi_1^+
                        \\
                        \tilde\psi_2 & \tilde\psi_2^+
                    \end{pmatrix}
                }{>}\raisebox{-20pt}{${}_{\frL^{(2)}}$}\ \coloneqq\ \int\rmd^dx\,\big\{\phi\tilde\phi^++\psi_1\tilde\psi_1^++\psi_2\tilde\psi_2^++\tilde\phi\phi^++\tilde\psi_1\psi_1^++\tilde\psi_2\psi_2^+\big\}
            \end{equation}
            for all $(\phi,\psi_1,\psi_2)\in\frL^{(2)}_1$ and $(\phi^+,\psi^+_1,\psi^+_2)\in\frL^{(2)}_2$. 
        \end{subequations}     
        
        \paragraph{Homotopy algebra for $S^{(c)}$.}
        The homotopy algebra $\frL^{(c)}$ corresponding to the action $S^{(c)}$ has underlying cochain complex
        \begin{subequations}\label{eq:L_infty_scal_hat}
            \begin{equation}
                \sfCh(\frL^{(c)})\ \coloneqq\ \left( 
                \begin{tikzcd}
                    \stackrel{\phi}{\Omega^0(\IM^d)}\arrow[r,"\wave"]  & \stackrel{\phi^+}{\Omega^0(\IM^d)}
                    \\[-0.5cm]
                    \stackrel{(\chi_1,\psi_1,\xi_1)}{\IR^3\otimes\Omega^0(\IM^d)}\arrow[start anchor=south east, end anchor={160},dr,"\sfid",pos=0.95,swap] & 
                    \stackrel{(\chi^+_1,\psi^+_1,\xi^+_1)}{\IR^3\otimes\Omega^0(\IM^d)}
                    \\[-0.5cm]
                    \underbrace{\spacer{2ex}\stackrel{(\chi_2,\psi_2,\xi_2)}{\IR^3\otimes\Omega^0(\IM^d)}}_{\eqqcolon\,\frL^{(c)}_1}\arrow[start anchor={20}, end anchor= south west,crossing over,ur,"\sfid",pos=0.95]
                    &
                    \underbrace{\spacer{2ex}\stackrel{(\chi^+_2,\psi^+_2,\xi^+_2)}{\IR^3\otimes\Omega^0(\IM^d)}}_{\eqqcolon\,\frL^{(c)}_2}
                \end{tikzcd}\right)\,,
            \end{equation}
            and the non-trivial higher products are the polarisations of 
            \begin{equation}
                \mu^{(c)}_2\left(
                \begin{pmatrix}
                    \phi & \phi^+ 
                    \\
                    \chi_1 & \chi_1^+
                    \\
                    \psi_1 & \psi_1^+
                    \\ 
                    \xi_1 & \xi_1^+
                    \\
                    \chi_2 & \chi_2^+
                    \\
                    \psi_2 & \psi_2^+
                    \\ 
                    \xi_2 & \xi_2^+
                \end{pmatrix},
                \begin{pmatrix}
                    \phi & \phi^+ 
                    \\
                    \chi_1 & \chi_1^+
                    \\
                    \psi_1 & \psi_1^+
                    \\ 
                    \xi_1 & \xi_1^+
                    \\
                    \chi_2 & \chi_2^+
                    \\
                    \psi_2 & \psi_2^+
                    \\ 
                    \xi_2 & \xi_2^+
                \end{pmatrix}\right)
                \ \coloneqq\ -\frac{2\lambda}{\sqrt{6!}}
                \begin{pmatrix}
                    0 & 2\phi\psi_1+2\phi \xi_2
                    \\
                    0 & 0
                    \\
                    0 & \phi^2
                    \\
                    0 & 0
                    \\
                    0 & 0
                    \\
                    0 & 0
                    \\
                    0 & \phi^2
                \end{pmatrix}
                -2
                \begin{pmatrix}
                    0 & \psi_2\chi_1+\chi_2\xi_1
                    \\
                    0 & \psi_2\phi
                    \\
                    0 & 0
                    \\ 
                    0 & \chi_2\phi
                    \\
                    0 & \phi\xi_1
                    \\
                    0 & \phi\chi_1
                    \\
                    0 & 0
                \end{pmatrix}
            \end{equation}
            for all $(\phi,\chi_1,\psi_1,\xi_1,\chi_2,\psi_2,\xi_2)\in\frL_1^{(c)}$ and $(\phi^+,\chi^+_1,\psi^+_1,\xi^+_1,\chi^+_2,\psi^+_2,\xi^+_2)\in\frL_2^{(c)}$. The metric structure is, once more, the evident pairing of fields and anti-fields.
        \end{subequations}     
        
        \subsection{Span of \texorpdfstring{$L_\infty$}{Linfty}-algebras}
        
        The four actions $S^{(b)}$, $S^{(1)}$, $S^{(2)}$, and $S^{(c)}$ correspond to the four homotopy algebras $\frL^{(b)}$, $\frL^{(1)}$, $\frL^{(2)}$, and $\frL^{(c)}$ which fit into 
        \begin{equation}\label{eq:diamondPhi6}
            \begin{tikzcd}
                & \frL^{(c)} \arrow[<->,dl] \arrow[<->,dr] &
                \\
                \frL^{(1)} \arrow[<->,dr] & & \frL^{(2)} \arrow[<->,dl] 
                \\
                & \frL^{(b)}
            \end{tikzcd}
        \end{equation}
        where the arrows indicate quasi-isomorphisms. Moreover any downwards arrow can be formulated as a homotopy transfer, as we shall show in the following. In conclusion, the upper half of the diamond~\eqref{eq:diamondPhi6} forms a span of $L_\infty$-algebras.
        
        \paragraph{Homotopy transfer $\frL^{(c)}\rightarrow\frL^{(1)}$.}
        This homotopy transfer starts from the special deformation retract 
        \begin{subequations}
            \begin{equation}
                \begin{tikzcd}
                    \ar[loop,out=160,in=200,distance=20,"\sfh" left] \big(\frL^{(c)},\mu^{(c)}_1\big)\arrow[r,shift left]{}{\sfp} & \big(\frL^{(1)},\mu^{(1)}_1\big) \arrow[l,hookrightarrow,shift left]{}{\sfe}
                \end{tikzcd}
            \end{equation}
            with the following embedding and projection maps and contracting homotopy:
            \begin{equation}
                \begin{gathered}
                    \sfe
                    \begin{pmatrix}
                        \phi & \phi^+
                        \\
                        \chi_1 & \chi_1^+
                        \\
                        \chi_2 & \chi_2^+
                    \end{pmatrix}
                    \coloneqq\
                    \begin{pmatrix}
                        \phi & \phi^+ 
                        \\
                        \chi_1 & \chi_1^+
                        \\
                        0 & 0
                        \\ 
                        0 & 0
                        \\
                        \chi_2 & \chi_2^+
                        \\
                        0 & 0
                        \\ 
                        0 & 0
                    \end{pmatrix}\,,~~~
                    \sfp
                    \begin{pmatrix}
                        \phi & \phi^+ 
                        \\
                        \chi_1 & \chi_1^+
                        \\
                        \psi_1 & \psi_1^+
                        \\ 
                        \xi_1 & \xi_1^+
                        \\
                        \chi_2 & \chi_2^+
                        \\
                        \psi_2 & \psi_2^+
                        \\ 
                        \xi_2 & \xi_2^+
                    \end{pmatrix}
                    \ \coloneqq\
                    \begin{pmatrix}
                        \phi & \phi^+
                        \\
                        \chi_1 & \chi_1^+
                        \\
                        \chi_2 & \chi_2^+
                    \end{pmatrix}\,,
                    \\
                    \sfh
                    \begin{pmatrix}
                        \phi & \phi^+ 
                        \\
                        \chi_1 & \chi_1^+
                        \\
                        \psi_1 & \psi_1^+
                        \\ 
                        \xi_1 & \xi_1^+
                        \\
                        \chi_2 & \chi_2^+
                        \\
                        \psi_2 & \psi_2^+
                        \\ 
                        \xi_2 & \xi_2^+
                    \end{pmatrix}
                    \ \coloneqq\ 
                    \begin{pmatrix}
                        0 & 0
                        \\
                        0 & 0
                        \\
                        \psi^+_2 & 0
                        \\ 
                        \xi^+_2 & 0
                        \\
                        0 & 0
                        \\
                        \psi_1^+ & 0
                        \\ 
                        \xi_1^+ & 0
                    \end{pmatrix}\,,
                \end{gathered}
            \end{equation}
        \end{subequations}
        so that
        \begin{equation}
            \sfid_{\frL^{(c)}}-\sfe\circ\sfp\ =\ \sfh\circ\mu^{(c)}_1+\mu^{(c)}_1\circ\sfh~,
        \end{equation}
        and the side conditions~\eqref{eq:HT_side_conditions} hold.
        
        The homotopy transfer, by formulas~\eqref{eq:transfer_formulas}, yields the new embedding
        \begin{subequations}
            \begin{equation}
                \sfE\,:\,\frL^{(1)}\ \rightarrow\ \frL^{(c)}
            \end{equation}
            with $\sfE_1\coloneqq\sfe$ and
            \begin{equation}
                \sfE_2\left(
                \begin{pmatrix}
                    \phi & \phi^+
                    \\
                    \chi_1 & \chi_1^+
                    \\
                    \chi_2 & \chi_2^+
                \end{pmatrix},
                \begin{pmatrix}
                    \phi & \phi^+
                    \\
                    \chi_1 & \chi_1^+
                    \\
                    \chi_2 & \chi_2^+
                \end{pmatrix}
                \right)\ \coloneqq\ 2
                \begin{pmatrix}
                    0 & 0
                    \\
                    0 & 0
                    \\
                    \phi\chi_1 & 0
                    \\ 
                    \frac{\lambda}{\sqrt{6!}}\phi^2 & 0
                    \\
                    0 & 0
                    \\
                    \frac{\lambda}{\sqrt{6!}}\phi^2 & 0
                    \\
                    \phi\chi_2 & 0
                \end{pmatrix}
            \end{equation}
        \end{subequations}
        being the only non-trivial components. As a consequence, the induced higher products~\eqref{eq:TransferredHigherProducts} are only non-trivial when $i=3$, and the single resulting higher product $\mu_3^{(1)}$ coincides with~\eqref{eq:mu_3_scal_1}. We conclude that there is a quasi-isomorphism between $\frL^{(c)}$ and $\frL^{(1)}$ that originates from a homotopy transfer.
        
        \paragraph{Homotopy transfer $\frL^{(c)}\rightarrow \frL^{(2)}$.}
        Let us now show the same for the quasi-isomorphism $\frL^{(c)}\rightarrow\frL^{(2)}$. Here, we consider a homotopy transfer starting from the special deformation retract 
        \begin{subequations}
            \begin{equation}
                \begin{tikzcd}
                    \ar[loop,out=160,in=200,distance=20,"\sfh" left] \big(\frL^{(c)},\mu^{(c)}_1\big)\arrow[r,shift left]{}{\sfp} & \big(\frL^{(2)},\mu^{(2)}_1\big)\arrow[l,hookrightarrow,shift left]{}{\sfe}
                \end{tikzcd}
            \end{equation}
            with the easily obtained maps
            \begin{equation}
                \begin{gathered}
                    \sfe
                    \begin{pmatrix}
                        \phi & \phi^+
                        \\
                        \psi_1 & \psi_1^+
                        \\
                        \psi_2 & \psi_2^+
                    \end{pmatrix}
                    \ \coloneqq\ \begin{pmatrix}
                        \phi & \phi^+ 
                        \\
                        \chi_1 & \chi_1^+
                        \\
                        \psi_1 & \psi_1^+
                        \\ 
                        \xi_1 & \xi_1^+
                        \\
                        \chi_2 & \chi_2^+
                        \\
                        \psi_2 & \psi_2^+
                        \\ 
                        \xi_2 & \xi_2^+
                    \end{pmatrix}\,,~~~
                    \sfp
                    \begin{pmatrix}
                        \phi & \phi^+ 
                        \\
                        \chi_1 & \chi_1^+
                        \\
                        \psi_1 & \psi_1^+
                        \\ 
                        \xi_1 & \xi_1^+
                        \\
                        \chi_2 & \chi_2^+
                        \\
                        \psi_2 & \psi_2^+
                        \\ 
                        \xi_2 & \xi_2^+
                    \end{pmatrix}
                    \ \coloneqq\
                    \begin{pmatrix}
                        \phi & \phi^+
                        \\
                        \psi_1 & \psi_1^+
                        \\
                        \psi_2 & \psi_2^+
                    \end{pmatrix}\,,~
                    \\
                    \sfh
                    \begin{pmatrix}
                        \phi & \phi^+ 
                        \\
                        \chi_1 & \chi_1^+
                        \\
                        \psi_1 & \psi_1^+
                        \\ 
                        \xi_1 & \xi_1^+
                        \\
                        \chi_2 & \chi_2^+
                        \\
                        \psi_2 & \psi_2^+
                        \\ 
                        \xi_2 & \xi_2^+
                    \end{pmatrix}
                    \ \coloneqq\
                    \begin{pmatrix}
                        0 & 0
                        \\
                        \chi^+_2 & 0
                        \\
                        0 & 0
                        \\ 
                        \xi^+_2 & 0
                        \\
                        \chi_1^+ & 0
                        \\
                        0 & 0
                        \\ 
                        \xi_1^+ & 0
                    \end{pmatrix}\,,
                \end{gathered}
            \end{equation}
        \end{subequations}
        and we have
        \begin{equation}
            \sfid_{\frL^{(c)}}-\sfe\circ\sfp\ =\ \sfh\circ\mu^{(c)}_1+\mu^{(c)}_1\circ\sfh~,
        \end{equation}
        together with the side conditions~\eqref{eq:HT_side_conditions}. 
        
        Here, formulas~\eqref{eq:transfer_formulas} lead to the embedding
        \begin{subequations}
            \begin{equation}
                \sfE\,:\,\frL^{(2)}\ \rightarrow\ \frL^{(c)}
            \end{equation}
            with $\sfE_1\coloneqq\sfe$ and
            \begin{equation}
                \begin{gathered}
                    \sfE_2\left(
                    \begin{pmatrix}
                        \phi & \phi^+
                        \\
                        \chi_1 & \chi_1^+
                        \\
                        \chi_2 & \chi_2^+
                    \end{pmatrix},
                    \begin{pmatrix}
                        \phi & \phi^+
                        \\
                        \chi_1 & \chi_1^+
                        \\
                        \chi_2 & \chi_2^+
                    \end{pmatrix}
                    \right)\ \coloneqq\ 
                    \begin{pmatrix}
                        0 & 0
                        \\
                        0 & 0
                        \\
                        0 & 0
                        \\
                        \frac{2\lambda}{\sqrt{6!}}\phi^2 & 0
                        \\
                        2\psi_2 \phi & 0
                        \\
                        0 & 0
                        \\
                        0 & 0
                    \end{pmatrix}\,,
                    \\
                    \sfE_3\left(
                    \begin{pmatrix}
                        \phi & \phi^+
                        \\
                        \chi_1 & \chi_1^+
                        \\
                        \chi_2 & \chi_2^+
                    \end{pmatrix},
                    \begin{pmatrix}
                        \phi & \phi^+
                        \\
                        \chi_1 & \chi_1^+
                        \\
                        \chi_2 & \chi_2^+
                    \end{pmatrix}
                    ,
                    \begin{pmatrix}
                        \phi & \phi^+
                        \\
                        \chi_1 & \chi_1^+
                        \\
                        \chi_2 & \chi_2^+
                    \end{pmatrix}
                    \right)\ \coloneqq\ 3
                    \begin{pmatrix}
                        0 & 0
                        \\
                        \frac{2\lambda}{\sqrt{6!}}\phi^3 & 0
                        \\ 
                        0 & 0
                        \\
                        0 & 0
                        \\
                        0 & 0
                        \\
                        0 & 0
                        \\
                        2\psi_2\phi^2 & 0 
                    \end{pmatrix}
                \end{gathered}
            \end{equation}
        \end{subequations}
        being the only non-trivial higher maps. The only non-trivial induced higher products~\eqref{eq:TransferredHigherProducts} resulting from~\eqref{eq:transfer_formulas} are then~\eqref{eq:L_infty_scal_2_products}. 
        
        \section{(Non-Abelian) T-duality of the principal chiral model}\label{sec:PCM}
        
        We now turn to a physically more interesting pair of quasi-isomorphic theories: the principal chiral model and its (non-Abelian) T-dual.
        
        \subsection{Homotopy algebraic formulation of the involved theories}
        
        We start by listing the theories involved in the (non-Abelian) T-duality, together with their homotopy algebraic formulation in terms of cyclic $L_\infty$-algebras.
        
        \paragraph{Principal chiral model.}
        We work on two-dimensional Minkowski space $\IM^2$ with metric $\eta=\diag(-1,1)$. We shall use the de Rham complex $(\Omega^\bullet(\IM^2),\rmd)$ together with the codifferential 
        \begin{equation}
            \rmd^\dagger\ \coloneqq\ {\star\rmd\star}~.
        \end{equation}
        We then have 
        \begin{equation}
            \rmd\rmd^\dagger+\rmd^\dagger\rmd\ =\ -\wave~,
        \end{equation}
        where $\wave$ is the d'Alembertian. We also note that
        \begin{equation}
            \star^2|_{\Omega^p(\IM^2)}\ =\ -(-1)^p\sfid_{\Omega^p(\IM^2)}~.
        \end{equation}
        
        The target space of the principal chiral model (PCM) is a Lie group $\sfG$ with Lie algebra $\frg$, and we assume that there is a non-degenerate, invariant, symmetric bilinear form $\inner{-}{-}_\frg$ on $\frg$. The kinematical data is then given by a smooth map 
        \begin{equation}
            g\,:\,\IM^2\ \rightarrow\ \sfG~,
        \end{equation}
        which we may parametrise as $g=\rme^\phi$ for $\phi\in\Omega^0(\IM^2,\frg)$. The pullback of the Maurer--Cartan form then reads as 
        \begin{equation}\label{eq:PCMcurrent}
            j\ \coloneqq\ g^{-1}\rmd g\ =\ \sum_{n=0}^\infty\frac{(-1)^n}{(n+1)!}\,\ad_\phi^n(\rmd\phi)
        \end{equation}
        with $\ad_\phi(-)\coloneqq[\phi,-]$, and the action of the PCM is given by 
        \begin{equation}\label{eq:action_PCM}
            S^{(1)}\ \coloneqq\ -\frac12\int\inner{j}{\star j}_\frg\ =\ -\sum_{n=0}^\infty\frac{1}{(2n+2)!}\int\inner{\rmd\phi}{\star\ad^{2n}_\phi(\rmd\phi)}_\frg~.
        \end{equation}
        In rewriting this action, we have noted that only even powers of $\ad_\phi$ will contribute after inserting~\eqref{eq:PCMcurrent} and then substituted the identity $\sum_{m=0}^{2n}\frac{(-1)^m}{(m+1)!(2n-m+1)!}=\frac{2}{(2n+2)!}$. 
        
        From the homotopy-algebraic perspective, the action~\eqref{eq:action_PCM} is given by a cyclic $L_\infty$-algebra $\frL^{(1)}$ with the underlying cochain complex
        \begin{subequations}\label{eq:L_infty_1}
            \begin{equation}
                \sfCh(\frL^{(1)})\ \coloneqq\ \left( 
                \begin{tikzcd}
                    \underbrace{\spacer{2ex}\stackrel{\phi}{\Omega^0(\IM^2,\frg)}}_{\eqqcolon\,\frL^{(1)}_1} \arrow[r,"\wave"] 
                    &
                    \underbrace{\spacer{2ex}\stackrel{\phi^+}{\Omega^0(\IM^2,\frg)}}_{\eqqcolon\,\frL^{(1)}_2}
                \end{tikzcd}\right)\,,
            \end{equation}
            non-trivial higher products that are obtained by polarising
            \begin{equation}\label{eq:higherProductsPCM}
                \tilde\mu^{(1)}_{2n+1}(\phi,\ldots,\phi)\ \coloneqq\ -\rmd^\dagger\ad^{2n}_\phi(\rmd\phi)
            \end{equation}
            for all $n\in\IN$, and the metric structure
            \begin{equation}\label{eq:cyclicStructurePCM}
                \inner{\phi}{\phi^+}_{\frL^{(1)}}\ \coloneqq\ \int\inner{\phi}{\star\phi^+}_\frg
            \end{equation}
            for all $\phi\in\frL^{(1)}_1$ and $\phi^+\in\frL^{(1)}_2$.
        \end{subequations}
        We stress that the simple polarisation of the expressions~\eqref{eq:higherProductsPCM} is not sufficient to render them cyclic, which has to be done separately. To this end, it is useful to note that  
        \begin{equation}
            \begin{aligned}
                &\int\inner{\rmd\phi}{\star\ad^{2n}_\phi(\rmd\phi)}_\frg
                \\
                &\kern1cm=\ \frac{1}{2n+2}\int\left\langle\phi,\star\left\{\sum_{m=0}^{2n-1}(-1)^i{\star[\ad^m_\phi(\rmd\phi),\star\ad^{2n-1-m}_\phi(\rmd\phi)]}+2\rmd^\dagger\ad^{2n}_\phi(\rmd\phi)\right\}\right\rangle_\frg
            \end{aligned}
        \end{equation}
        as a short calculation reveals. Upon combining this with~\eqref{eq:action_PCM}, we read off the equivalent higher products 
        \begin{equation}\label{eq:alternativeHigherProductsPCM}
            \tilde\mu^{(1)}_{2n+1}(\phi,\ldots,\phi)\ \coloneqq\ -\frac{1}{2n+2}\left\{\sum_{m=0}^{2n-1}(-1)^m{\star[\ad^m_\phi(\rmd\phi),\star\ad^{2n-1-m}_\phi(\rmd\phi)]}+2\rmd^\dagger\ad^{2n}_\phi(\rmd\phi)\right\}.
        \end{equation}
        The polarisation of these higher products are indeed cyclic with respect to~\eqref{eq:cyclicStructurePCM}. 
        
        \paragraph{Gauged principal chiral model.}
        Following~\cite{Rocek:1991ps,delaOssa:1992vci}, the first step to T-dualise the PCM is to gauge a normal Lie subgroup $\sfH$ of $\sfG$ that corresponds to the directions that we wish to T-dualise. Let $\frh$ be the Lie algebra of $\sfH$. The gauging is implemented by introducing an $\frh$-valued connection one-form $\omega\in\Omega^1(\IM^2,\frh)$ so that the current~\eqref{eq:PCMcurrent} generalises to
        \begin{equation}
            j_\omega\ \coloneqq\ g^{-1}\omega g+g^{-1}\rmd g~.
        \end{equation}
        Evidently, with $F_\omega\coloneqq \rmd\omega+\frac12[\omega,\omega]$ and $F_{j_\omega}\coloneqq \rmd j_\omega+\frac12[j_\omega,j_\omega]$, we have $F_{j_\omega}=g^{-1}F_\omega g$ implying equivalence of the flatness of $j_\omega$ and $\omega$. Furthermore, $j_\omega$ is invariant under the local $\sfH$-action
        \begin{subequations}\label{eq:HgaugeTransformations}
            \begin{equation}
                g\ \mapsto\ h^{-1}g
                \eand
                \omega\ \mapsto\ h^{-1}\omega h+h^{-1}\rmd h
            \end{equation}
            for any smooth map $h:\IM^2\rightarrow\sfH$. To implement the flatness $F_\omega=0$, we introduce a Lagrange multiplier $\Lambda\in\Omega^0(\IM^2,\frh)$ subject to the local $\sfH$-action
            \begin{equation}
                \Lambda\ \mapsto\ h^{-1}\Lambda h~.
            \end{equation}
        \end{subequations}
        Hence, the gauged PCM action is
        \begin{equation}\label{eq:gaugedActionPCM}
            S^{(1)}_{\rm gauged}\ \coloneqq\ -\frac12\int\inner{j_\omega}{\star j_\omega}_\frg+\int\inner{\Lambda}{F_\omega}_\frg\ =\ -\frac12\int\inner{j_\omega}{\star j_\omega}_\frg+\int\inner{\tilde\Lambda}{F_{j_\omega}}_\frg
        \end{equation}
        with $\tilde\Lambda\coloneqq g^{-1}\Lambda g$. This last form of the action makes the $\sfH$-gauge invariance manifest since both $j_\omega$ and $\tilde\Lambda$ are invariant under the transformations~\eqref{eq:HgaugeTransformations}. Upon integrating out $\Lambda$ and gauge-fixing $\omega$ to zero, we recover the action~\eqref{eq:action_PCM} of the PCM.
        
        \pagebreak
        
        Furthermore, the explicit form of the local $\sfH$-action on the field $\phi$ follows from the Baker--Campbell--Hausdorff formula. Parametrising $h=\rme^c$ for $c\in\Omega^0(\IM^2,\frh)$, we find to linear order in $c$
        \begin{equation}
            \phi\ \mapsto\ \phi-c+\frac12[\phi,c]-\frac{1}{12}[\phi,[\phi,c]]+\cdots\ =\ \phi-\sum_{n=0}^\infty\frac{B_n^-}{n!}\ad^n_\phi(c)~,
        \end{equation}
        where $B_n^-$ are the Bernoulli numbers with the convention that $B^-_1=-\tfrac12$. Altogether, the BV action for the gauged PCM then reads as
        \begin{equation}\label{eq:action_correspondence_theory}
            \begin{aligned}
                S^{(c)}\ &\coloneqq\ S^{(1)}_{\rm gauged}+\int\left\{\inner{\omega^+}{\star(\rmd c+[\omega,c])}_\frg-\sum_{n=0}^\infty\frac{B_n^-}{n!}\inner{\phi^+}{\ad^n_\phi(c)}_\frg+\frac12\inner{c^+}{[c,c]}_\frg\right\}
                \\
                &\,=\ \int\left\{-\frac12\inner{\omega}{\star\omega}_\frg-\sum_{n=0}^\infty\frac{1}{(n+1)!}\inner{\omega}{\star\ad^n_\phi(\rmd\phi)}_\frg-\sum_{n=0}^\infty\frac{1}{(2n+1)!}\inner{\rmd\phi}{\ad_\phi^{2n}(\rmd\phi)}_\frg\right.
                \\
                &\kern1.5cm+\left.\inner{\Lambda}{F_\omega}_\frg+\inner{\omega^+}{\star(\rmd c+[\omega,c])}_\frg-\sum_{n=0}^\infty\frac{B_n^-}{n!}\inner{\phi^+}{\ad^n_\phi(c)}_\frg+\frac12\inner{c^+}{[c,c]}_\frg\right\}.
            \end{aligned}
        \end{equation}
        This action corresponds to a cyclic $L_\infty$-algebra $\frL^{(c)}$ with the underlying cochain complex
        \begin{subequations}\label{eq:MotherL_inftyAlgebra}   
            \begin{equation}\label{eq:MotherDifferentialComplex}
                \sfCh(\frL^{(c)})\ \coloneqq\ \left( 
                \begin{tikzcd}
                    \stackrel{c}{\Omega^0(\IM^2,\frh)}\arrow[r,"-\id"]\arrow[rd,"\rmd"]
                    & 
                    \stackrel{\phi}{\Omega^0(\IM^2,\frg)}\arrow[r,"\wave"]\arrow[start anchor=south east, end anchor={150},dr,"-\rmd",pos=0.8,swap,outer sep=-1.5pt]
                    & 
                    \stackrel{\phi^+}{\Omega^0(\IM^2,\frg)}\arrow[r,"\id"]
                    &
                    \stackrel{c^+}{\Omega^0(\IM^2,\frh)}
                    \\
                    & 
                    \stackrel{\omega}{\Omega^1(\IM^2,\frh)}\arrow[r,"-\sfid",outer sep=-1.5pt]\arrow[start anchor={30}, end anchor={200},crossing over,ur,"-\rmd^\dagger",pos=0.8,outer sep=-2pt]\arrow[start anchor=south east, end anchor={150},dr,"-\star\rmd",pos=0.85,swap,outer sep=-1.5pt]  
                    &
                    \stackrel{\omega^+}{\Omega^1(\IM^2,\frh)}\arrow[ru,"-\rmd^\dagger"]
                    &
                    \\
                    \underbrace{\spacer{2ex}\phantom{\Omega^0(\IM^2,\frg)}}_{\eqqcolon\,\frL^{(c)}_0}
                    & 
                    \underbrace{\spacer{2ex}\stackrel{\Lambda}{\Omega^0(\IM^2,\frh)}}_{\eqqcolon\,\frL^{(c)}_1} \arrow[start anchor={30}, end anchor= south west,crossing over,ur,"\star\rmd",pos=0.85,outer sep=-1.5pt]
                    &
                    \underbrace{\spacer{2ex}\stackrel{\Lambda^+}{\Omega^0(\IM^2,\frh)}}_{\eqqcolon\,\frL^{(c)}_2} 
                    &
                    \underbrace{\spacer{2ex}\phantom{\Omega^0(\IM^2,\frh)}}_{\eqqcolon\,\frL^{(c)}_3}
                \end{tikzcd}\right)\,,
            \end{equation}
            and non-trivial higher products are given by
            \begin{equation}\label{eq:higherProductGaugedBCMBV}
                \begin{aligned}
                    \mu^{(c)}_2(c,c)|_c\ &\coloneqq\ [c,c]~,
                    \\
                    \mu^{(c)}_n(\phi,\ldots,\phi,c)|_\phi\ &\coloneqq\ -B^-_{n-1}\ad_\phi^{n-1}(c)~,  
                    \\
                    \mu^{(c)}_2(\omega,c)|_\omega\ &\coloneqq\ [\omega,c]~,
                    \\
                    \mu^{(c)}_n(\phi,\ldots,\phi)|_{\phi^+}\ &\coloneqq\ 
                    \begin{cases}
                        0 & \efor n\in2\IN
                        \\
                        -\rmd^\dagger\ad^{n-1}_\phi(\rmd\phi) & \eelse
                    \end{cases}~,
                    \\
                    \mu^{(c)}_n(\phi,\ldots,\phi)|_{\omega^+}\ &\coloneqq\ -\ad^{n-1}_\phi(\rmd \phi)~,
                    \\
                    \mu^{(c)}_n(\phi,\ldots,\phi,\omega)|_{\phi^+}\ &\coloneqq\ (-1)^{n-1}\rmd^\dagger\ad^{n-1}_\phi(\omega)~,
                    \\
                    \mu^{(c)}_2(\omega,\omega)|_{\Lambda^+}\ &\coloneqq\ -{\star[\omega,\omega]}~,
                    \\
                    \mu^{(c)}_2(\omega,\Lambda)|_{\omega^+}\ &\coloneqq\ {\star[\omega,\Lambda]}~,
                    \\
                    \mu^{(c)}_n(\phi,\ldots,\phi,\phi^+,c)|_{\phi^+}\ &\coloneqq\ -B_{n-1}^-\ad_{\phi^+}(\ad_\phi^{n-2}(c))~,
                    \\
                    \mu_2(\omega^+,c)|_{\omega^+}\ &\coloneqq\ [\omega^+,c]~,
                    \\
                    \mu^{(c)}_n(\phi,\ldots,\phi,\phi^+)|_{c^+}\ &\coloneqq\ (-1)^{n-1}B_{n-1}^-\ad_\phi^{n-1}(\phi^+)~,
                    \\
                    \mu^{(c)}_2(\omega,\omega^+)|_{c^+}\ &\coloneqq\ [\omega,\omega^+]~,
                \end{aligned}
            \end{equation}
            where all the fields are elements of the evident subspaces of $\frL^{(c)}$. The metric structure is defined by 
            \begin{equation}\label{eq:cyclicStructureGaugedBCMBV}
                \begin{aligned}
                    &\inner{c+\phi+\omega+\Lambda}{c^++\phi^++\omega^++\Lambda^+}_{\frL^{(c)}}
                    \\
                    &\kern1cm\coloneqq\ \int\big\{\inner{c}{\star c^+}_\frg+\inner{\phi}{\star\phi^+}_\frg+\inner{\omega}{\star\omega^+}_\frg+\inner{\Lambda}{\star\Lambda^+}_\frg\big\}\,.
                \end{aligned}
            \end{equation}
        \end{subequations}
        Note that the general higher products follow from polarisation of~\eqref{eq:higherProductGaugedBCMBV} and cyclification via~\eqref{eq:cyclicStructureGaugedBCMBV} similar to~\eqref{eq:alternativeHigherProductsPCM}. 
        
        \paragraph{T-dual principal chiral model.}
        If we integrate out the gauge potential $\omega$ from the action~\eqref{eq:gaugedActionPCM}, we obtain the action of the T-dual model~\cite{Rocek:1991ps,delaOssa:1992vci}. In particular, the equation of motion for $\omega$ is
        \begin{equation}
            \star j_\omega\ =\ \rmd\tilde\Lambda+[j_\omega,\tilde\Lambda]~.
        \end{equation}
        This is an algebraic equation for $j_\omega$ which has the solution
        \begin{equation}
            j_\omega\ =\ -\frac12\frac{1}{1-\ad_{\tilde\Lambda}}(\rmd\tilde\Lambda+{\star\rmd\tilde\Lambda})+\frac12\frac{1}{1+\ad_{\tilde\Lambda}}(\rmd\tilde\Lambda-{\star\rmd\tilde\Lambda})~.
        \end{equation}
        Upon substituting this into~\eqref{eq:gaugedActionPCM}, we obtain the T-dual action
        \begin{equation}
            \begin{aligned}
                S^{(2)}\ &\coloneqq\ -\frac12\int_\Sigma\left\langle\rmd\tilde\Lambda,\frac{1}{1-\ad_{\tilde\Lambda}}(\rmd\tilde\Lambda+\star\rmd\tilde\Lambda)\right\rangle_\frh
                \\
                &\,=\ -\frac12\int\left\{\inner{\rmd\tilde\Lambda}{\star\rmd\tilde\Lambda}_\frh+\sum_{n=1}^\infty\inner{\rmd\tilde\Lambda}{\ad_{\tilde\Lambda}^{2n-1}(\rmd\tilde\Lambda)}_\frh+\sum_{n=1}^\infty\inner{\rmd\tilde\Lambda}{\star\ad_{\tilde\Lambda}^{2n}(\rmd\tilde\Lambda)}_\frh\right\}.
            \end{aligned}
        \end{equation}
        The corresponding homotopy algebra $\frL^{(2)}$ has the underlying cochain complex 
        \begin{subequations}\label{eq:L_infty_2}
            \begin{equation}
                \sfCh(\frL^{(2)})\ \coloneqq\ \left( 
                \begin{tikzcd}
                    \underbrace{\spacer{2ex}\stackrel{\tilde\Lambda}{\Omega^0(\IM^2,\frh)}}_{\eqqcolon\,\frL^{(2)}_1}\arrow[r,"\wave"] 
                    &
                    \underbrace{\spacer{2ex}\stackrel{\tilde\Lambda^+}{\Omega^0(\IM^2,\frh)}}_{\eqqcolon\,\frL^{(2)}_2}
                \end{tikzcd}
                \right)\,,
            \end{equation}
            higher products defined by
            \begin{equation}\label{eq:higherProductsTDualPCM}
                \mu^{(2)}_n(\tilde\Lambda,\ldots,\tilde\Lambda)\ \coloneqq\ -\tfrac{(n+1)!}{2}\rmd^\dagger\ad^{n-1}_{\tilde\Lambda}(\star^{n-1}\rmd\tilde\Lambda)~,
            \end{equation} 
            and the metric structure
            \begin{equation}\label{eq:cyclicStructureTDualPCM}
                \inner{\tilde\Lambda}{\tilde\Lambda^+}_{\frL^{(2)}}\ \coloneqq\ \int\inner{\tilde\Lambda}{\star\tilde\Lambda^+}_\frh
            \end{equation}
            for all $\tilde\Lambda\in\frL^{(2)}_1$ and $\tilde\Lambda^+\in\frL^{(2)}_2$.
        \end{subequations}
        Again, the general form of the higher products is obtained from the polarisation of~\eqref{eq:higherProductsTDualPCM} followed by the cyclification with respect to~\eqref{eq:cyclicStructureTDualPCM}. Hence, similar to~\eqref{eq:alternativeHigherProductsPCM}, we may equivalently take
        \begin{equation}\label{eq:alternativeHigherProductsTDualPCM}
            \tilde\mu^{(2)}_n(\tilde\Lambda,\ldots,\tilde\Lambda)\ \coloneqq\ -\frac{n!}{2}\left\{\sum_{m=0}^{n-2}(-1)^m{\star\big[\ad_{\tilde\Lambda}^m(\rmd\tilde\Lambda),\ad_{\tilde\Lambda}^{n-2-m}(\star^n\rmd\tilde\Lambda)\big]}+2\rmd^\dagger\ad_{\tilde\Lambda}^{n-1}(\star^{n-1}\rmd\tilde\Lambda)\right\}
        \end{equation}
        instead of~\eqref{eq:higherProductsTDualPCM} whose polarisation is directly cyclic.
        
        \subsection{Span of \texorpdfstring{$L_\infty$}{Linfty}-algebras}
        
        We now describe the homotopy transfers realising the quasi-isomorphisms that link the PCM to its T-dual model and produce a span of $L_\infty$-algebras
        \begin{equation}
            \begin{tikzcd}
                & \frL^{(c)} \arrow[<->,dl] \arrow[<->,dr] &
                \\
                \frL^{(1)} & & \frL^{(2)} 
            \end{tikzcd}
        \end{equation}
        We start with the simpler transfer from $\frL^{(c)}$ to $\frL^{(2)}$.
        
        \paragraph{Homotopy transfer $\frL^{(c)}\rightarrow\frL^{(2)}$.}
        Between the differential complexes underlying $\frL^{(c)}$ and $\frL^{(2)}$, we have the following special deformation retract: 
        \begin{subequations}
            \begin{equation}
                \begin{tikzcd}
                    \ar[loop,out=160,in=200,distance=20,"\sfh" left] (\frL^{(c)},\mu_1^{(c)})\arrow[r,shift left]{}{\sfp} & (\frL^{(2)},\mu^{(2)}_1) \arrow[l,hookrightarrow,shift left]{}{\sfe}
                \end{tikzcd}
            \end{equation}
            with 
            \begin{equation}\label{eq:contractingHomotopyPCM1}
                \begin{aligned}
                    \sfe
                    \begin{pmatrix}
                        \tilde\Lambda & \tilde\Lambda^+
                    \end{pmatrix}
                    \ &\coloneqq\ 
                    \begin{pmatrix} 
                        0 & 0 & 0 & 0
                        \\
                        & -{\star\rmd\tilde\Lambda} & 0 
                        \\
                        & -\tilde\Lambda & -\tilde\Lambda^+
                    \end{pmatrix}\,,
                    \\
                    \sfp
                    \begin{pmatrix}
                        c & \phi & \phi^+ & c^+
                        \\
                        & \omega & \omega^+ 
                        \\
                        & \Lambda & \Lambda^+
                    \end{pmatrix}
                    \ &\coloneqq\ 
                    \begin{pmatrix} 
                        -\Lambda & -(\Lambda^+-{\star\rmd\omega^+})
                    \end{pmatrix}\,,
                    \\
                    \sfh
                    \begin{pmatrix}
                        c & \phi & \phi^+ & c^+
                        \\
                        & \omega & \omega^+ 
                        \\
                        & \Lambda & \Lambda^+
                    \end{pmatrix}
                    \ &\coloneqq\ 
                    \begin{pmatrix}
                        -\phi & 0 & c^+ & 0
                        \\
                        & -\omega^+ & 0 
                        \\
                        & 0 & 0
                    \end{pmatrix}\,,
                \end{aligned}
            \end{equation}
        \end{subequations}        
        where the positions indicate the subspaces of the complexes in which the expressions take values, as displayed in~\eqref{eq:L_infty_2} and~\eqref{eq:MotherDifferentialComplex}. In particular, we have~\eqref{eq:deformation_retract_2}, and the side conditions~\eqref{eq:HT_side_conditions} are satisfied as well.
        
        We note that in the formulas~\eqref{eq:transfer_formulas}, the arguments of the higher products are always applied to images of $\sfE_n$, which, in turn are images of either $\sfe$ or $\sfh$. For degree~$1$, these images are contained in the subspace $\frL^{(c)}_{1,\omega}\oplus\frL^{(c)}_{1,\Lambda}$, and it thus suffices to consider the higher brackets in $\frL^{(c)}$ restricted to these subspaces,
        \begin{equation}
            \begin{aligned}
                \mu^{(c)}_2(\omega,\omega)|_{\Lambda^+}\ &=\ -{\star[\omega,\omega]}~,
                \\
                \mu^{(c)}_2(\omega,\Lambda)|_{\omega^+}\ &=\ {\star[\omega,\Lambda]}~.
            \end{aligned}
        \end{equation}
        Moreover, the only higher products we need to evaluate are $\mu^{(2)}_n(\tilde\Lambda,\ldots,\tilde\Lambda)$, and we can therefore simplify the formulas~\eqref{eq:transfer_formulas} as 
        \begin{subequations}
            \begin{equation}\label{eq:EnMunInTermsOfFnPCM1}
                \sfE_n(\tilde\Lambda,\ldots,\tilde\Lambda)\ =\ -\sfh(\sfF_n(\tilde\Lambda,\ldots,\tilde\Lambda))
                \eand
                \mu^{(2)}_n(\tilde\Lambda,\ldots,\tilde\Lambda)\ =\ \sfp(\sfF_n(\tilde\Lambda,\ldots,\tilde\Lambda))
            \end{equation}
            with
            \begin{equation}\label{eq:defOfFn}
                \sfF_n(\tilde\Lambda,\ldots,\tilde\Lambda)\ \coloneqq\ \frac{1}{2!}\sum_{\substack{k_1+k_2=n\\k_1,k_2\geq 1}}\binom{n}{k_1}\mu_2^{(c)}(\sfE_{k_1}(\tilde\Lambda,\ldots,\tilde\Lambda),\sfE_{k_2}(\tilde\Lambda,\ldots,\tilde\Lambda)) 
            \end{equation}
        \end{subequations}
        for all $n>1$. 
        
        \pagebreak
        
        We shall now prove inductively that
        \begin{equation}
            \begin{aligned}
                \sfE_n(\tilde\Lambda,\ldots,\tilde\Lambda)|_\Lambda\ &=\ -\delta_{n1}\tilde\Lambda~,
                \\
                \sfE_n(\tilde\Lambda,\ldots,\tilde\Lambda)|_\omega\ &=\ -n!\,\ad_{\tilde\Lambda}^{n-1}({\star^n\rmd\tilde\Lambda})
            \end{aligned}
        \end{equation}
        for all $n\in\IN$. These relations evidently hold for $n=1$. Suppose now that they hold for $1,\ldots,n-1$ with $n>2$. Then,
        \begin{subequations}
            \begin{equation}
                \begin{aligned}
                    \sfF_n(\tilde\Lambda,\ldots,\tilde\Lambda)|_{\Lambda^+}\ &=\ \frac{1}{2!}\sum_{\substack{k_1+k_2=n\\k_1,k_2\geq 1}}\binom{n}{k_1}\mu_2^{(c)}(\sfE_{k_1}(\tilde\Lambda,\ldots,\tilde\Lambda)|_\omega,\sfE_{k_2}(\tilde\Lambda,\ldots,\tilde\Lambda)|_\omega)
                    \\
                    &=\ -\frac{n!}2\sum_{\substack{k_1+k_2=n\\k_1,k_2\geq 1}}{\star\big[\ad_{\tilde\Lambda}^{k_1-1}({\star^{k_1}\rmd\tilde\Lambda}),\ad_{\tilde\Lambda}^{k_2-1}({\star^{k_2}\rmd\tilde\Lambda})\big]}
                    \\
                    &=\ -\frac{n!}2\sum_{\substack{k_1+k_2=n\\k_1,k_2\geq 1}}(-1)^{k_1}{\star\big[\ad_{\tilde\Lambda}^{k_1-1}(\rmd\tilde\Lambda),\ad_{\tilde\Lambda}^{k_2-1}({\star^n\rmd\tilde\Lambda})\big]}
                    \\
                    &=\ \frac{n!}2\sum_{\substack{k_1+k_2=n-2\\k_1,k_2\geq 0}}(-1)^{k_1}{\star\big[\ad_{\tilde\Lambda}^{k_1}(\rmd\tilde\Lambda),\ad_{\tilde\Lambda}^{k_2}({\star^n\rmd\tilde\Lambda})\big]}
                    \\
                    &=\ \frac{n!}2\sum_{m=0}^{n-2}(-1)^m{\star\big[\ad_{\tilde\Lambda}^m(\rmd\tilde\Lambda),\ad_{\tilde\Lambda}^{n-2-m}({\star^n\rmd\tilde\Lambda})\big]}~,
                \end{aligned}
            \end{equation}
            where in the third step, we have used $[\alpha,{\star^k\beta}]=(-1)^k[{\star^k\alpha},\beta]$ for any two Lie-algebra-valued differential one-forms $\alpha$ and $\beta$ and in the last step, we have used the identity $\sum_{\substack{j+k=i\\j,k\geq 0}}a_jb_k=\sum_{j=0}^ia_ib_{j-i}$. Likewise,
            \begin{equation}
                \begin{aligned}
                    \sfF_n(\tilde\Lambda,\ldots,\tilde\Lambda)|_{\Lambda^+}\ &=\ \frac{1}{2!}\sum_{\substack{k_1+k_2=n\\k_1,k_2\geq 1}}\binom{n}{k_1}\mu_2^{(c)}(\sfE_{k_1}(\tilde\Lambda,\ldots,\tilde\Lambda)|_\omega,\sfE_{k_2}(\tilde\Lambda,\ldots,\tilde\Lambda)|_\Lambda)
                    \\
                    &=\ \binom{n}{n-1}\mu_2^{(c)}(\sfE_{n-1}(\tilde\Lambda,\ldots,\tilde\Lambda)|_\omega,\sfE_1(\tilde\Lambda,\ldots,\tilde\Lambda)|_\Lambda)
                    \\
                    &=\ n!\star{\big[\ad_{\tilde\Lambda}^{n-2}(\star^{n-1}\rmd\tilde\Lambda),\tilde\Lambda\big]}
                    \\
                    &=\ -n!\,\ad_{\tilde\Lambda}^{n-1}(\star^n\rmd\tilde\Lambda)~.
                \end{aligned}
            \end{equation}
        \end{subequations}
        Hence, using~\eqref{eq:contractingHomotopyPCM1} and~\eqref{eq:EnMunInTermsOfFnPCM1}, we immediately find
        \begin{subequations}
            \begin{equation}
                \sfE_n(\tilde\Lambda,\ldots,\tilde\Lambda)|_\Lambda\ =\ -\sfh(\sfF_n(\tilde\Lambda,\ldots,\tilde\Lambda))|_\Lambda\ =\ 0
            \end{equation}
            for all $n>1$. Likewise,
            \begin{equation}
                \sfE_n(\tilde\Lambda,\ldots,\tilde\Lambda)|_\omega\ =\ -\sfh(\sfF_n(\tilde\Lambda,\ldots,\tilde\Lambda)|_{\omega^+})\ =\ \sfF_n(\tilde\Lambda,\ldots,\tilde\Lambda)|_{\omega^+}\ =\ -n!\,\ad_{\tilde\Lambda}^{n-1}(\star^n\rmd\tilde\Lambda)~.
            \end{equation}
        \end{subequations}
        This completes the proof.
        
        Now, with~\eqref{eq:contractingHomotopyPCM1} and~\eqref{eq:EnMunInTermsOfFnPCM1}, we find
        \begin{equation}
            \begin{aligned}
                &\mu^{(2)}_n(\tilde\Lambda,\ldots,\tilde\Lambda)|_{\tilde\Lambda^+}\ =\ \sfp(\sfF_n(\tilde\Lambda,\ldots,\tilde\Lambda))
                \\
                &\hspace{1cm}=\ -\sfF_n(\tilde\Lambda,\ldots,\tilde\Lambda)|_{\Lambda^+}+{\star\rmd\sfF_n(\tilde\Lambda,\ldots,\tilde\Lambda)|_{\omega^+}}
                \\
                &\hspace{1cm}=\ -\frac{n!}{2}\left\{\sum_{m=0}^{n-2}(-1)^m{\star\big[\ad_{\tilde\Lambda}^m(\rmd\tilde\Lambda),\ad_{\tilde\Lambda}^{n-2-m}(\star^n\rmd\tilde\Lambda)\big]}+2\rmd^\dagger\ad_{\tilde\Lambda}^{n-1}(\star^{n-1}\rmd\tilde\Lambda)\right\},
            \end{aligned}
        \end{equation}
        which are precisely the higher products~\eqref{eq:alternativeHigherProductsTDualPCM} on $\frL^{(2)}$.
        
        \paragraph{Homotopy transfer from $\frL^{(c)}\rightarrow\frL^{(1)}$.} 
        The homotopy equivalence between~\eqref{eq:MotherL_inftyAlgebra} and~\eqref{eq:L_infty_1} is somewhat more complicated, and we need to consider the involved functions carefully, distinguishing between fields with null and non-null momenta.
        
        We assume that all fields are bounded at infinity, so that we can expand them in terms of plane waves $x\mapsto\rme^{\rmi p\cdot x}$ with momentum $p$. Because the differential $\mu_1$ in the complex~\eqref{eq:MotherDifferentialComplex} preserves momenta, we can decompose $\Omega^1(\IM^2)$ as
        \begin{equation}
            \Omega^1(\IM^2)\ \cong\ \Omega^1_\text{e}(\IM^2)\oplus \Omega^1_\text{c}(\IM^2)\oplus \Omega^1_\text{ec}(\IM^2)\oplus \Omega^1_\text{r}(\IM^2)\oplus\Omega^1_\text{cm}(\IM^2)~,
        \end{equation}
        where we have introduced the subspaces
        \begin{equation}
            \begin{aligned}
                \Omega^1_\text{e}(\IM^2)\ &\coloneqq\ \{\text{one-forms with $p^2\neq0$ and spanned by}~\rmd x^\mu p_\mu\rme^{\rmi p\cdot x}\}~,
                \\
                \Omega^1_\text{c}(\IM^2)\ &\coloneqq\ \{\text{one-forms with $p^2\neq0$ and spanned by}~\rmd x^\mu\eps_{\mu\nu}p^\nu\rme^{\rmi p\cdot x}\}~,
                \\
                \Omega^1_\text{ec}(\IM^2)\ &\coloneqq\ \{\text{one-forms with $p^2=0$ and $p\neq 0$ and spanned by}~\rmd x^\mu p_\mu\rme^{\rmi p\cdot x}\}~,
                \\
                \Omega^1_\text{r}(\IM^2)\ &\coloneqq\ \{\text{one-forms with $p^2=0$ and $p\neq 0$ and spanned by}~\rmd x^\mu\delta_{\mu\nu}p^\nu\rme^{\rmi p\cdot x}\}~,
                \\
                \Omega^1_\text{cm}(\IM^2)\ &\coloneqq\ \{\text{one-forms with $p=0$ and spanned by}~\rmd x^\mu\}~,
            \end{aligned}
        \end{equation}
        where $\eps_{\mu\nu}$ is the Levi-Civita symbol and $\delta_{\mu\nu}$ the Kronecker symbol, respectively. Elements of $\Omega^1_\text{e}(\IM^2)$ are exact and elements of $\Omega^1_\text{c}(\IM^2)$ are coexact. Furthermore, while elements of $\Omega^1_\text{ec}(\IM^2)$ are both closed and coclosed since $\rmd x^\mu p_\mu=\pm\rmd x^\mu\eps_{\mu\nu}p^\nu$ for $p_0=\pm p_1$, elements of $\Omega^1_\text{r}(\IM^2)$ are neither closed nor coclosed.  We also have
        \begin{equation}
            \star\,:\,\Omega^1_\text{e}(\IM^2)\ \rightarrow\ \Omega^1_\text{c}(\IM^2)
            \eand
            \star\,:\,\Omega^1_\text{c}(\IM^2)\ \rightarrow\ \Omega^1_\text{e}(\IM^2)~,
        \end{equation}
        and elements in $\Omega^1_\text{ec}(\IM^2)$ and $\Omega^1_\text{r}(\IM^2)$ with definite momentum $p$ are either self-dual or anti-self-dual, depending on the sign in $p_0=\pm p_1$. 
        
        \pagebreak
        
        In the following, we shall denote by $\Pi_\text{e}$, $\Pi_\text{c}$, $\Pi_\text{ec}$, $\Pi_\text{r}$, and $\Pi_\text{cm}$ the projectors onto those subspaces. We also introduce the map $P$ which makes Poincar\'e's lemma concrete and is defined as follows
        \begin{equation}
            \begin{aligned}
                P\,:\,\Omega^1_\text{e}(\IM^2)\oplus\Omega^1_\text{ec}(\IM^2)\oplus\Omega^2(\IM^2)\ &\rightarrow\ \Omega^0(\IM^2)\oplus \Omega^1_\text{c}(\IM^2)\oplus\Omega^1_\text{r}(\IM^2)~,
                \\
                \rmd x^\mu p_\mu \rme^{\rmi p\cdot x}\ &\mapsto\ -\rmi\rme^{\rmi px}~,
                \\
                \tfrac12\eps_{\mu\nu}\rmd x^\mu\wedge\rmd x^\nu\rme^{\rmi p\cdot x}\ &\mapsto\ 
                \rmi\begin{cases}
                    \frac{1}{p^2}\rmd x^\mu\eps_{\mu\nu}p^\nu\rme^{\rmi p\cdot x} & \efor p^2\neq 0~,
                    \\
                    -\frac{1}{2p_0p_1}\rmd x^\mu\delta_{\mu\nu}p^\nu\rme^{\rmi p\cdot x} & \efor p^2=0~,~~p\neq 0~,
                    \\
                    0 & \efor p=0~.
                \end{cases}
            \end{aligned}
        \end{equation}
        We notice that for differential forms $\alpha_0\in \Omega^0(\IM^2)$, $\alpha_{1,\text{e/ec}} \in \Omega^1_\text{e}(\IM^2)\oplus \Omega^1_\text{ec}(\IM^2)$, $\alpha_{1,\text{c/r}} \in\Omega^{1}_\text{c}(\IM^2)\oplus\Omega^{1}_\text{r}(\IM^2)$, $\alpha_{1,\text{e/r}}\in \Omega^{1}_\text{e}(\IM^2)\oplus\Omega^{1}_\text{r}(\IM^2)$, and $\alpha_2\in \Omega^2(\IM^2)$, we have 
        \begin{equation}
            \begin{aligned}
                \rmd P(\alpha_{1,\text{e/ec}})\ &=\ \alpha_{1,\text{e/ec}}~,
                \\
                \rmd P(\alpha_2)\ &=\ \alpha_2-\Pi_\text{cm}(\alpha_2)~,
                \\
                -{\star\rmd P}({\star\alpha_0})\ &=\ \alpha_0-\Pi_\text{cm}(\alpha_0)~,
                \\
                P(\rmd\alpha_0)\ &=\ \alpha_0-\Pi_\text{cm}(\alpha_0)~,
                \\
                P\rmd(\alpha_{1,\text{c/r}})\ &=\ \alpha_{1,\text{c/r}}~,
                \\
                {\star P}({\star(-\rmd^\dagger\alpha_{1,\text{e/r}})})\ &=\ \alpha_{1,\text{e/r}}~.
            \end{aligned}
        \end{equation}
        
        The underlying graded vector space of the $L_\infty$-algebra $\frL^{(c)}$ given in~\eqref{eq:MotherL_inftyAlgebra} thus decomposes as
        \begin{subequations}
            \begin{equation} \label{eq:MotherDifferentialComplex2}
                \frL^{(c)}\ \cong\ \left( 
                \begin{tikzcd}[row sep=-0.1cm]
                    \stackrel{c}{\Omega^0(\IM^2,\frh)} 
                    & 
                    \stackrel{\phi}{\Omega^0(\IM^2,\frg)} 
                    &[20pt]
                    \stackrel{\phi^+}{\Omega^0(\IM^2,\frg)}
                    &
                    \stackrel{c^+}{\Omega^0(\IM^2,\frh)}
                    \\
                    & \oplus & \oplus
                    \\
                    & 
                    \stackrel{\omega_\text{e}}{\Omega^1_\text{e}(\IM^2,\frh)} 
                    &
                    \stackrel{\omega^+_\text{e}}{\Omega^1_\text{e}(\IM^2,\frh)}
                    &
                    \\
                    & \oplus & \oplus
                    \\
                    & 
                    \stackrel{\omega_\text{c}}{\Omega^1_\text{c}(\IM^2,\frh)} 
                    &
                    \stackrel{\omega^+_\text{c}}{\Omega^1_\text{c}(\IM^2,\frh)}
                    &
                    \\
                    & \oplus & \oplus
                    \\
                    & 
                    \stackrel{\omega_\text{ec}}{\Omega^1_\text{ec}(\IM^2,\frh)} 
                    &
                    \stackrel{\omega^+_\text{ec}}{\Omega^1_\text{ec}(\IM^2,\frh)}
                    &
                    \\
                    & \oplus & \oplus
                    \\
                    & 
                    \stackrel{\omega_\text{r}}{\Omega^1_\text{r}(\IM^2,\frh)} 
                    &
                    \stackrel{\omega^+_\text{r}}{\Omega^1_\text{r}(\IM^2,\frh)}
                    &
                    \\
                    & \oplus & \oplus
                    \\
                    & 
                    \stackrel{\omega_\text{cm}}{\Omega^1_\text{cm}(\IM^2,\frh)} 
                    &
                    \stackrel{\omega^+_\text{cm}}{\Omega^1_\text{cm}(\IM^2,\frh)}
                    &
                    \\
                    & \oplus & \oplus
                    \\
                    \underbrace{\spacer{2ex}\phantom{\Omega^0(\IM^2,\frh)}}_{=\,\frL^{(c)}_0}& 
                    \underbrace{\spacer{2ex}\stackrel{\Lambda}{\Omega^0(\IM^2,\frh)}}_{\cong\,\frL^{(c)}_1} 
                    &
                    \underbrace{\spacer{2ex}\stackrel{\Lambda^+}{\Omega^0(\IM^2,\frh)}}_{\cong\,\frL^{(c)}_2} 
                    &
                    \underbrace{\spacer{2ex}\phantom{\Omega^0(\IM^2,\frh)}}_{=\,\frL^{(c)}_3}
                \end{tikzcd}\right)~,
            \end{equation}
            on which we have the differential
            \begin{equation}
            \resizebox{0.98\hsize}{!}{$
                \begin{aligned}
                    &\mu_1
                    \begin{pmatrix}
                        c & \phi & \phi^+ & c^+
                        \\
                        & \omega_\text{e} & \omega^+_\text{e} 
                        \\
                        & \omega_\text{c} & \omega^+_\text{c} 
                        \\
                        & \omega_\text{ec} & \omega^+_\text{ec} 
                        \\
                        & \omega_\text{r} & \omega^+_\text{r} 
                        \\
                        & \omega_\text{cm} & \omega^+_\text{cm} 
                        \\
                        & \Lambda & \Lambda^+
                    \end{pmatrix}
                    &=\ 
                    \begin{pmatrix} 
                        0 & -c & \wave\phi-\rmd^\dagger(\omega_\text{e}+\omega_\text{r}) & \phi^+-\rmd^\dagger(\omega^+_\text{e}+\omega^+_\text{r})
                        \\
                        & \Pi_\text{e}(\rmd c) & -\Pi_\text{e}(\rmd\phi)-\omega_\text{e}
                        \\
                        & 0 & -\omega_\text{c}+\Pi_\text{c}({\star\rmd\Lambda})
                        \\
                        & \Pi_\text{ec}(\rmd c) & -\Pi_\text{ec}(\rmd\phi)-\omega_\text{ec}+\Pi_\text{ec}({\star\rmd\Lambda})
                        \\
                        & 0 & -\omega_\text{r}
                        \\
                        & 0 & -\omega_\text{cm}
                        \\
                        & 0 & -{\star\rmd(\omega_\text{c}+\omega_\text{r})}
                    \end{pmatrix}\,.
                \end{aligned}
            $}
            \end{equation}
            It is not hard to see that
            \begin{equation}
                \begin{gathered}
                    \inner{\omega_\text{e}}{\omega^+_\text{c}}_{\frL^{(c)}}\ =\ \inner{\omega_\text{c}}{\omega^+_\text{e}}_{\frL^{(c)}}\ =\ \inner{\omega_\text{ec}}{\omega^+_\text{ec}}_{\frL^{(c)}}\ =\ \inner{\omega_\text{r}}{\omega^+_\text{r}}_{\frL^{(c)}}\ =\ 0~.
                \end{gathered}
            \end{equation}
        \end{subequations}
        
        The deformation retract can then be constructed in two steps, following the physical intuition: in a first step, we integrate out $\Lambda$ and in a second step, we gauge trivialise the remaining connection form. We can then use formula~\eqref{eq:combination_deformation_retracts} to combine both.
        
        The result is the special deformation retract
        \begin{subequations}
            \begin{equation}
                \begin{tikzcd}
                    \ar[loop,out=160,in=200,distance=20,"\sfh" left] (\frL^{(c)},\mu_1^{(c)})\arrow[r,shift left]{}{\sfp} & (\frL^{(1)},\mu^{(1)}_1)\arrow[l,hookrightarrow,shift left]{}{\sfe}
                \end{tikzcd}
            \end{equation}
            with 
            \begin{equation}\label{eq:contractingHomotopyPCM2}
                \begin{aligned}
                    \sfp
                    \begin{pmatrix}
                        c & \phi & \phi^+ & c^+
                        \\
                        & \omega_\text{e} & \omega^+_\text{e} 
                        \\
                        & \omega_\text{c} & \omega^+_\text{c} 
                        \\
                        & \omega_\text{ec} & \omega^+_\text{ec} 
                        \\
                        & \omega_\text{r} & \omega^+_\text{r} 
                        \\
                        & \omega_\text{cm} & \omega^+_\text{cm} 
                        \\
                        & \Lambda & \Lambda^+
                    \end{pmatrix}
                    \ &\coloneqq\
                    \begin{pmatrix}
                        \phi-\Pi_\text{cm}(\phi+\Lambda)+P(\omega_\text{e}+\omega_\text{ec}) & \phi^++S(\Lambda^+)+\Pi_\text{cm}(\Lambda^+)
                    \end{pmatrix},
                    \\
                    \sfe
                    \begin{pmatrix} 
                        \phi & \phi^+
                    \end{pmatrix}
                    \ &\coloneqq\ 
                    \begin{pmatrix}
                        0 & \phi & \phi^+-\Pi_\text{cm}(\phi^+) & 0
                        \\
                        & 0 & -\Pi_\text{e}(\star P(\star \phi^+))
                        \\
                        & 0 & 0
                        \\
                        & 0 & 0 
                        \\
                        & 0 & -\Pi_\text{r}(\star P(\star \phi^+))
                        \\
                        & 0 & 0
                        \\
                        & \Pi_\text{cm}(\phi)+P(\star \Pi_\text{ec}(\rmd\phi)) & \Pi_\text{cm}(\phi^+)
                    \end{pmatrix}\,,
                    \\
                    \sfh
                    \begin{pmatrix}
                        c & \phi & \phi^+ & c^+
                        \\
                        & \omega_\text{e} & \omega^+_\text{e} 
                        \\
                        & \omega_\text{c} & \omega^+_\text{c} 
                        \\
                        & \omega_\text{ec} & \omega^+_\text{ec} 
                        \\
                        & \omega_\text{r} & \omega^+_\text{r} 
                        \\
                        & \omega_\text{cm} & \omega^+_\text{cm} 
                        \\
                        & \Lambda & \Lambda^+
                    \end{pmatrix}
                    \\
                    &\kern-3.5cm\coloneqq\ 
                    \begin{pmatrix}
                        \Pi_\text{cm}(\Lambda-\phi)+P(\omega_\text{e}+\omega_\text{ec}) & 0 & \Pi_\text{cm}(c^+) & 0
                        \\
                        & 0 & \Pi_\text{e}(\star P(\star c^+))
                        \\
                        & \Pi_\text{c}(P{\star\Lambda^+}) & 0
                        \\
                        & 0 & 0
                        \\
                        & \Pi_\text{r}(P{\star\Lambda^+}) & \Pi_\text{r}(\star P(\star c^+))
                        \\
                        & -\omega_\text{cm}^+ & 0
                        \\
                        & P \star (\omega^+_\text{c}+\omega^+_\text{ec}+\Pi_\text{c}(P{\star\Lambda^+})) & -\Pi_\text{cm}(c^+)
                    \end{pmatrix}\,,
                \end{aligned}
            \end{equation}
            where we also used the map
            \begin{equation}
                S\,:\,\Omega^p(\IM^2)\ \rightarrow\ \Omega^p(\IM^2)
            \end{equation}
            which vanishes off-shell and inverts the sign of on-shell forms, depending on the momentum,
            \begin{equation}
                S(\rme^{\rmi p\cdot x})\ \coloneqq\ 
                \begin{cases}
                    \mp\rme^{\rmi p\cdot x} & p_0=\pm p_1\neq0
                    \\
                    0 & \mbox{else}
                \end{cases}~.
            \end{equation}
        \end{subequations}
        One straightforwardly checks that we have~\eqref{eq:deformation_retract_2}, and the side conditions~\eqref{eq:HT_side_conditions} are satisfied as well.
        
        It remains to show that the homotopy transfer indeed reproduces the higher products of $\frL^{(1)}$. Considering formulas~\eqref{eq:transfer_formulas}, we note the following. The embedding $\sfE_1=\sfe$ of $\frL^{(1)}$ into $\frL^{(c)}$ will map a field $\phi\in \frL^{(1)}_1$ to field components $\phi$ and $\Lambda$ in $\frL^{(c)}_1$. The only interactions between these are the interactions between $\phi$-components, which are given by the cyclified and polarized versions of
        \begin{equation}
            \begin{aligned}
                \mu^{(c)}_n(\phi,\ldots,\phi)|_{\phi^+}\ &\coloneqq\ 
                \begin{cases}
                    0 & \efor n\in2\IN
                    \\
                    -\rmd^\dagger\ad^{n-1}_\phi(\rmd\phi) & \eelse
                \end{cases}~,
                \\
                \mu^{(c)}_n(\phi,\ldots,\phi)|_{\omega^+}\ &\coloneqq\ -\ad^{n-1}_\phi(\rmd \phi)~.
            \end{aligned}
        \end{equation}
        Note that the $\omega^+$-component has no constant part, as the derivative of the functions that we are considering (i.e.~bounded at infinity) either vanishes or is non-constant. So applying $\sfh$ to the result will only produce a field component $\Lambda$ in $\frL^{(c)}_1$. In summary, the only arguments ever entering the higher products in the homotopy transfer will be the $\phi$- and $\Lambda$-components of $\frL^{(c)}_1$. The only non-trivial higher products with these arguments, however, are the ones with all arguments being $\phi$-components. The latter exclusively arise from the direct embedding via $\sfE_1=\sfe$. The final projector $\sfp$ is only non-trivial on the component fields $\phi^+$ and $\Lambda^+$, and therefore the homotopy transfer is just a pullback of the higher product defined by 
        \begin{equation}
            \mu^{(c)}_n(\phi,\ldots,\phi)|_{\phi^+}\ \coloneqq\ 
            \begin{cases}
                0 & \efor n\in2\IN~,
                \\
                -\rmd^\dagger\ad^{n-1}_\phi(\rmd\phi) & \eelse~,
            \end{cases}
        \end{equation}
        which reproduces the higher products on $\frL^{(1)}$.
        
        \section{Penrose--Ward transform}\label{sec:twistors}
        
        The purpose of this section is to briefly describe yet another example of spans of $L_\infty$-algebras, arising in the context of the Penrose--Ward transform.
        
        \paragraph{Penrose--Ward transform.}
        A number of gauge field equations can be written as flatness conditions on certain subspaces of space-time. The most prominent example is the instanton or self-dual Yang--Mills equations on $\IR^4$, which correspond to flatness of the connection along self-dual two-planes in $\IR^4$~\cite{Ward:1977ta}. Another important example is $\caN=3$ super Yang--Mills theory on $\IR^4$, which amounts to flatness along super light lines in $\IR^{4|12}$~\cite{Witten:1978xx}. 
        
        For such gauge field equations, one considers a double fibration of manifolds
        \begin{equation}\label{eq:DFPW}
            \begin{tikzcd}
                & F\arrow[dl,"\pi_1" swap]\arrow[dr,"\pi_2"]
                \\
                Z & & M
            \end{tikzcd}
        \end{equation}
        where $M$ is space-time, $Z$ is the twistor space, and $F$ is the correspondence space. In particular, by virtue of this double fibration, we have a geometric correspondence between points $x\in M$ and subspaces $\pi_1(\pi_2^{-1}(x))\subseteq Z$ and points $z\in Z$ and subspaces $\pi_2(\pi_1^{-1}(z))\subseteq M$, respectively. Moreover, in may interesting cases both $Z$ and subspaces $\pi_1(\pi_2^{-1}(x))\subseteq Z$ are complex manifolds.
        
        The Penrose--Ward transform is now the map between equivalence classes of holomorphic principal $\sfG$-bundles over $Z$, holomorphically trivial when restricted to the subspaces $\pi_1(\pi_2^{-1}(x)\subseteq Z$, to equivalence classes of holomorphic principal $\sfG$-bundles over $M$ equipped with a holomorphic connection that is flat on all the subspaces $\pi_2(\pi_1^{-1}(z))\subseteq M$. This flatness encodes the field configurations of the gauge field equations one wishes to study such as the aforementioned instanton equations~\cite{Ward:1977ta}. For more examples, see e.g.~\cite{Manin:1988ds,Ward:1990vs,Mason:1991rf}.
        
        In the Dolbeault picture, such holomorphic $\sfG$-principal bundles on $Z$ can be described by smooth complex $\sfG$-principal bundles equipped with a $(0,1)$-connection locally described by a $\frg$-valued $(0,1)$-forms $A^{0,1}$ subject to the holomorphic Chern--Simons equation
        \begin{equation}
            \bar\partial A^{0,1}+\tfrac12[A^{0,1},A^{0,1}]\ =\ 0~.
        \end{equation}
        In the process of the Penrose--Ward transform, these $(0,1)$-forms on $Z$ are mapped to relative one-forms $A_{\pi_1}$ along the fibres $\pi_1$ on $F$ that are relatively flat
        \begin{equation}
            \rmd_{\pi_1}A_{\pi_1}+\tfrac12[A_{\pi_1},A_{\pi_1}]\ =\ 0~.
        \end{equation}
        These relative one-forms can then be naturally pushed down to one-forms on $M$, and, in turn, the relative flatness equation becomes the relevant field equation on $M$.
        
        \paragraph{On-shell correspondence.}
        The double fibration~\eqref{eq:DFPW} already suggest a span of $L_\infty$-algebras, which, contrary to our previous cases, only works on-shell. Explicitly, we have the following picture
        \begin{equation}\label{eq:DFPWL}
            \begin{tikzcd}
                & \frL_F^\text{flat}\arrow[dl,"\sfp_1" swap]\arrow[dr,"\sfp_2"]
                \\
                \frL_P^\text{flat} & & \frL_M^\text{flat}
            \end{tikzcd}
        \end{equation}       
        where all $L_\infty$-algebras are concentrated in degrees $0$ and $1$, with the ghosts parametrising gauge transformations in degrees $0$ and
        \begin{equation}
            \begin{aligned}
                \frL_{M,1}^\text{flat}~&:~\mbox{one-forms that are solutions to the gauge field equations under consideration}
                \\
                \frL_{F,1}^\text{flat}~&:~\mbox{relative one-forms solutions which are relatively flat} 
                \\
                \frL_{Z,1}^\text{flat}~&:~\mbox{holomorphic $(0,1)$-forms solutions which are holomorphically flat} 
            \end{aligned}
        \end{equation}
        The Penrose--Ward transform establishes a quasi-isomorphism between all of these $L_\infty$-algebras. Moreover, the projections $\sfp_1$ and $\sfp_2$ merely amount to integrating out different degrees of gauge redundancy. 
        
        We stress that although there are evident completions of the $L_\infty$-algebras appearing in~\eqref{eq:DFPWL} to off-shell versions, the Penrose--Ward correspondence, and hence the span of $L_\infty$-algebras~\eqref{eq:DFPWL} does \emph{not} extend to those. The problem is that the gauge transformation necessary for translating relative one-forms $A_{\pi_1}$ on $F$ to one-forms $\pi_1^*A^{(0,1)}$ on $F$ that arise as pullbacks of one-forms $A^{(0,1)}$ on $Z$ only exists for flat such connections. In the following, we shall construct an off-shell example, which exists in a particular case.
        
        \paragraph{Real instantons.}
        Consider the special case of the $\caN=4$ supersymmetric instanton equations on real Euclidean $\IR^4$, cf.~\cite{Siegel:1992xp,Devchand:1992st}. Here, the double fibration~\eqref{eq:DFPW} collapses to a single fibration
        \begin{equation}\label{eq:single_fibration}
            Z^{3|4}\xleftarrow{~\cong~}\IR^{4|8}\times\IC P^1\xrightarrow{~~~}\IR^{4|8}~,
        \end{equation}
        where $Z^{3|4}$ is the total space of the rank $(2|4)$ holomorphic vector bundle $\IC^{2|4}\otimes \caO(1)\rightarrow \IC P^1$, where $\caO(1)$ denotes the complex line bundle over $\IC P^1$ of first Chern class $1$. The twistor space $Z^{3|4}$ comes with a holomorphic volume form $\Omega^{3|4,0|0}$~\cite{Witten:2003nn}, which allows us to write down the holomorphic Chern--Simons action
        \begin{equation}
            S_\text{hCS}\ \coloneqq\ \int \Omega^{3|4,0|0}\wedge\left\{\tfrac12\inner{A^{0,1}}{\bar\partial A^{0,1}}_\frg+\tfrac{1}{3!}\inner{A^{0,1}}{[A^{0,1},A^{0,1}]}_\frg\right\},
        \end{equation}
        where $A^{0,1}$ is a gauge-Lie algebra-valued $(0,1)$-form on $Z^{3|4}$ with purely holomorphic dependence on the fermionic coordinates and no anti-holomorphic fermionic directions.        
        
        It is well-known that holomorphic Chern--Simons theory on $Z^{3|4}$ is quasi-isomorphic to $\caN=4$ supersymmetric self-dual Yang--Mills theory given by the Siegel action~\cite{Siegel:1992xp}; see~\cite{Mason:2005zm,Boels:2006ir} and also~\cite{Wolf:2010av}. Both the holomorphic Chern--Simons action and the Siegel action can be extended to evident BV actions, and the corresponding $L_\infty$-algebras $\frL_{Z^{3|4}}$ and $\frL_{\IR^{4|8}}$ are quasi-isomorphic. Moreover, this quasi-isomorphism is a two-step homotopy transfer, cf.~\cite{Mason:2005zm,Boels:2006ir}, see also~\cite{Wolf:2010av}. In a first step, we use the contracting homotopy with
        \begin{equation}
            \sfh_1\ =\ \bar\partial^\dagger_{\IC P^1}
        \end{equation}
        the adjoint of the Dolbeault operator, restricted to $\pi_2^{-1}(x)\cong\IC P^1$ for all $x\in\IR^4$ to impose the space-time gauge. In a second step, we use a second homotopy transfer to integrate out all auxiliary fields, which leaves us with the space-time BV fields in $\frL_{\IR^{4|8}}$. These homotopy transfers are then concatenated as explained in~\eqref{eq:combination_deformation_retracts}. This quasi-isomorphism of $L_\infty$-algebras has recently been used in the context of colour--kinematics duality~\cite{Borsten:2022vtg} to derive kinematic Lie algebras from twistor spaces.
        
        \paragraph{Span of $L_\infty$-algebras with mini-twistors.}
        As explained in detail in~\cite{Popov:2005uv}, the single fibration~\eqref{eq:single_fibration} is expanded into a double fibration again when considering its dimensional reduction to three space-time dimensions. Explicitly, $\IR^{4|8}$ is reduced to $\IR^{3|8}$, but the twistor space $Z^{2|4}$ for the description of supersymmetric monopoles becomes a supersymmetric generalisation of the mini-twistor space introduced in~\cite{Hitchin:1982gh}, and $Z^{2|4}$ is the total space of the rank $(1|4)$-vector bundle $\caO(2)\oplus\IC^{0|4}\otimes\caO(1)\IC P^1$. We end up with the double fibration
        \begin{equation}\label{eq:double_2}
            \begin{tikzcd}
                & \IR^{3|8}\times \IC P^1\arrow[dl,"\pi_1" swap]\arrow[dr,"\pi_2"]
                \\
                Z^{2|4} & & \IR^{3|8}
            \end{tikzcd}
        \end{equation}
        This, in turn, induces a span of $L_\infty$-algebras
        \begin{equation}\label{eq:span_2}
            \begin{tikzcd}
                & \frL_{\IR^{3|8}\times\IC P^1}\arrow[dl,"\sfp_1" swap]\arrow[dr,"\sfp_2"]
                \\
                \frL_{Z^{2|4}} & & \frL_{\IR^{3|8}}
            \end{tikzcd}
        \end{equation}
        which are the $L_\infty$-algebras of the BV extensions of the following field theories:
        \begin{equation}
            \begin{aligned}
                \frL_{\IR^{3|8}}~&:~\mbox{supersymmetric monopole theory}
                \\
                \frL_{\IR^{3|8}\times\IC P^1}~&:~\mbox{partially holomorphic Chern--Simons theory as defined in~\cite{Popov:2005uv}}
                \\
                \frL_{Z^{2|4}}~&:~\mbox{holomorphic BF theory as defined in~\cite{Popov:2005uv}}
            \end{aligned}
        \end{equation}
        Evidently, these $L_\infty$-algebras are quasi-isomorphic, and in the span of $L_\infty$-algebras~\eqref{eq:span_2}, the homotopy transfer $\sfp_2$ is given by a real dimensional reduction of the homotopy transfer from $\frL_{Z^{3|4}}$ to $\frL_{\IR^{4|8}}$, while the homotopy transfer $\sfp_1$ amounts to a push-forward, as explained in~\cite{Popov:2005uv}.
        
        \paragraph{Summary.}
        Altogether, we have described an interesting and fully off-shell example of a span of $L_\infty$-algebras arising in the context of twistor spaces and the Penrose--Ward transform. We note that this example can be extended to arbitrary amount of supersymmetry at the cost of the action principles. At the level of $L_\infty$-algebras, we merely loose the metric structure. All structures of homotopy transfer, and, in particular the analogues of the span~\eqref{eq:span_2}, however, remain valid.
        
    \end{body}    
    
\end{document}